\documentclass[11pt,fleqn,a4paper]{article} 

\usepackage{amsmath,amssymb,amsthm,enumerate} 
\usepackage[mathscr]{eucal}
\usepackage{graphicx}

\usepackage{hyperref}
\hypersetup{colorlinks, linkcolor=blue, citecolor=blue, urlcolor=blue}

\allowdisplaybreaks

\newcommand{\p}{\partial}

\newcommand{\const}{{\rm const}}

\newcommand{\rank}{\mathop{\rm rank}\nolimits}
\newcommand{\ord}{\mathop{\rm ord}\nolimits}

\newcommand{\spanindex}{{\mbox{\tiny$\langle\,\rangle$}}}

\newtheorem{theorem}{Theorem}

\newtheorem{corollary}[theorem]{Corollary}
\newtheorem{conjecture}[theorem]{Conjecture}
\newtheorem{proposition}[theorem]{Proposition}
\newtheorem*{proposition*}{Proposition}
{\theoremstyle{definition}

\newtheorem{remark}[theorem]{Remark}

}

\newcommand{\todo}[1][\null]{\ensuremath{\clubsuit}}

\newcommand{\noprint}[1]{}

\newcounter{mcasenum}

\begin{document}
\par\noindent {\LARGE\bf
Zeroth-order conservation laws\\ of two-dimensional shallow water equations\\ with variable bottom topography\par}

\vspace{5mm}\par\noindent{\large
Alexander Bihlo$^\dag$ and Roman O.\ Popovych$^\ddag$
}

\vspace{5mm}\par\noindent{\it
$^\dag$\,Department of Mathematics and Statistics, Memorial University of Newfoundland,\\
$\phantom{^\dag}$\,St.\ John's (NL) A1C 5S7, Canada
}

\vspace{2mm}\par\noindent{\it
$^\ddag$\,Fakult\"at f\"ur Mathematik, Universit\"at Wien, Oskar-Morgenstern-Platz 1, A-1090 Wien, Austria
\\
$\phantom{^\ddag}$Institute of Mathematics of NAS of Ukraine, 3 Tereshchenkivska Str., 01024 Kyiv, Ukraine
}

\vspace{6mm}\par\noindent
E-mails:
abihlo@mun.ca,
rop@imath.kiev.ua

\vspace{7mm}\par\noindent\hspace*{8mm}\parbox{140mm}{\small
We classify zeroth-order conservation laws of systems from the class
of two-dimensional shallow water equations with variable bottom topography
using an optimized version of the method of furcate splitting. 
The classification is carried out up to equivalence generated by the equivalence group of this class.
We find additional point equivalences between some of the listed cases 
of extensions of the space of zeroth-order conservation laws,
which are inequivalent up to transformations from the equivalence group. 
Hamiltonian structures of systems of shallow water equations are used
for relating the classification of zeroth-order conservation laws of these systems to 
the classification of their Lie symmetries. 
We also construct generating sets of such conservation laws under action of Lie symmetries.
}\par\vspace{5mm}

\noprint{

76-XX   Fluid mechanics {For general continuum mechanics, see 74Axx, or other parts of 74-XX}
 76Mxx  Basic methods in fluid mechanics [See also 65-XX]
  76M60   Symmetry analysis, Lie group and algebra methods
86-XX  Geophysics [See also 76U05, 76V05]
 86Axx	Geophysics [See also 76U05, 76V05]
  86A05  Hydrology, hydrography, oceanography [See also 76Bxx, 76E20, 76Q05, 76Rxx, 76U05]
37Kxx  Infinite-dimensional Hamiltonian systems [See also 35Axx, 35Qxx] 
  37K05   Hamiltonian structures, symmetries, variational principles, conservation laws 
35-XX   Partial differential equations
  35A30   Geometric theory, characteristics, transformations [See also 58J70, 58J72]
  35B06   Symmetries, invariants, etc.

}

\section{Introduction}

The shallow water equations play a distinguished role amongst models in geophysical fluid mechanics. 
They are used both as a starting model for constructing novel numerical schemes 
that should eventually be implemented for the full set 
of three-dimensional hydro-thermodynamical equations governing the evolution of the atmosphere--ocean system, 
as well as a practical model for simulating the evolution of water waves 
on the open ocean as needed for tsunami pre-warnings. 
See \cite{aech15a,bihl12By, brec19a,brec18a,cott14a,dele10a,flye12a,salm07a, tito97Ay,tito95a,tito98a,wan13Ay} 
and references therein for several results for both types of uses of the shallow water equations.

The wide areas of application of the shallow water equations justify exhaustive theoretical analysis as well, 
which has been the subject of various research programs over the past several decades. 
These considerations include the computation of exact solutions~\cite{cama19a,carr58a,levi89By,thac81a}, 
Lie symmetries \cite{akse16a,bila06Ay,ches09Ay,ches11a,chir14a,levi89By,siri16a,szat14a}, 
conservation laws \cite{akse16a,akse20a,cama19a,cava82a,wan13Ay} 
and the derivation of Hamiltonian structures \cite{cama19a,cava82a,salm98a,salm07a,shep90a}, 
just to name a few. 

We remind in more details some references on conservation laws of shallow water equations 
as this is the subject of the present paper. 
Conservation laws of the system of one-dimensional shallow water equations with flat bottom topography
were considered in~\cite{cava82a} jointly with a Hamiltonian formulation for this system found in~\cite[Section~II.6.5]{mani78a};
see also \cite[Section~2.2]{cama19a}. 
The case of linearly sloping bottom profile was discussed in~\cite{akyi82a,akyi87a}. 
Zeroth-order conservation laws for systems 
of one-dimensional shallow water equations with variable bottom topography in Eulerian coordinates 
were computed in~\cite{akse16a} depending on bottom topography profiles. 
First-order conservation laws for the one-dimensional 
variable bottom topography in Lagrangian variables were constructed in~\cite{akse20a} 
and then considered within the framework of Noether's theorem. 
The conservation laws of the two-dimensional case have previously been considered 
in \cite{egge94a,hydo05a,mont01a} as well as in \cite{salm98a,shep90a}, 
and have been used to construct conservative numerical schemes in \cite{salm07a,somm09a,wan13Ay}. 
All of these considerations, however, focus on the case of flat bottom topography.

We note that knowing conservation laws of systems of differential equations in geophysical fluid mechanics 
is of paramount importance for developing numerical models of these systems. 
As there are very few exact solutions known to the increasingly complex models 
arising in this field against which numerical models could be benchmarked, 
conservation laws provide critical information for assessing the quality of these numerical models. 

In the present paper, we classify zeroth-order conservation laws of systems from the class
of two-dimensional shallow water equations with variable bottom topography,
which are of the form
\begin{gather}\label{eq:2DSWEs}
\begin{split}
&u_t+uu_x+vu_y+h_x-b_x=0,\\
&v_t+uv_x+vv_y+h_y-b_y=0,\\
&h_t+(uh)_x+(vh)_y=0.
\end{split}
\end{gather}
Here
$(u,v)$ is the horizontal fluid velocity averaged over the height of the fluid column,
$h$ is the thickness of a fluid column,
$b=b(x,y)$ is the bottom topography measured downward with respect to a fixed reference level, 
which plays the role of a parameter function, 
and the gravitational acceleration is set to be equal 1 in dimensionless units.
See Figure~\ref{picture} for a graphical representation of the above quantities. 
Physically, systems from the class~\eqref{eq:2DSWEs} can be derived by depth integrating the (incompressible) Euler equations 
under the assumption that the horizontal scale is much larger than the vertical scale, 
with the fluid under consideration being in hydrostatic balance~\cite{vall17a}.

\begin{figure}
\centering
\includegraphics[width=.85\linewidth]{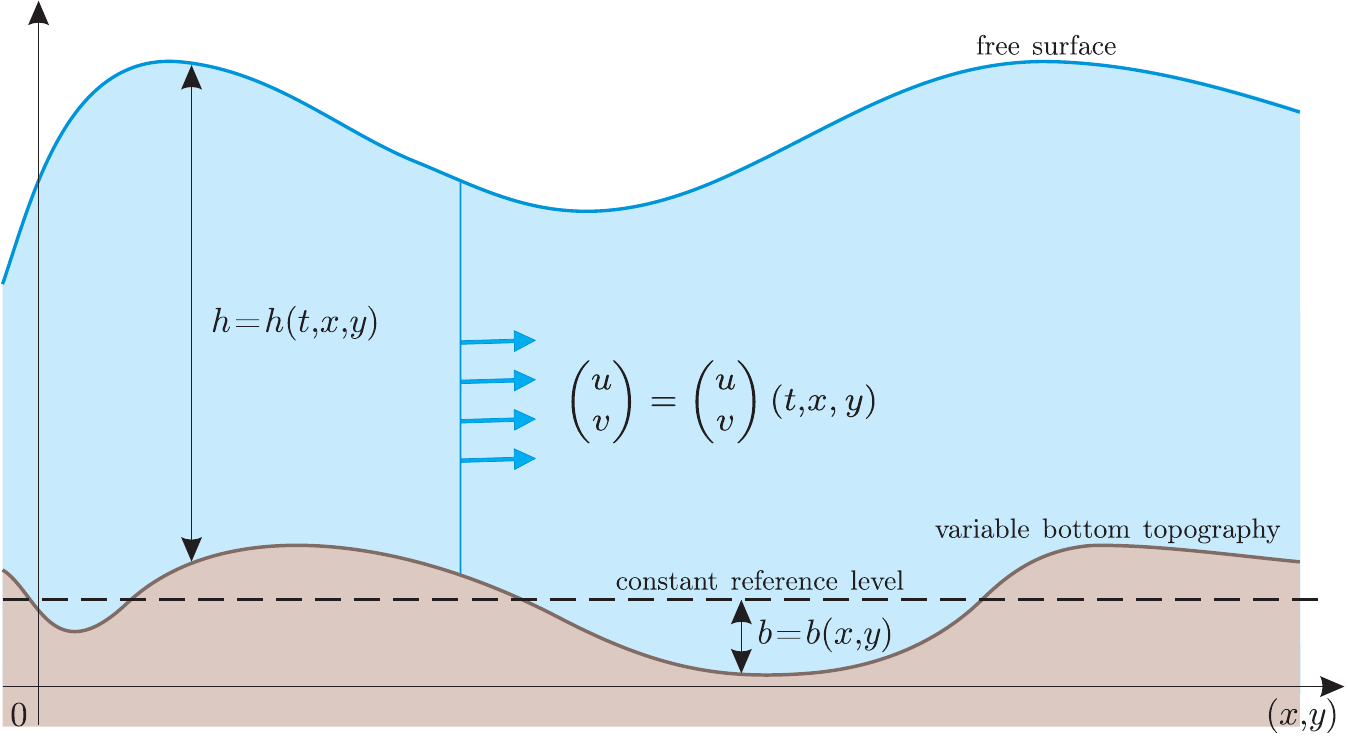}
\caption{The shallow water model.}
\label{picture}
\end{figure}

Systems of shallow water equations in the form~\eqref{eq:2DSWEs} are of vital relevance to modeling tsunamis, 
and it is for this case that the form of the bottom topography profile is especially relevant. 
Knowing specific forms of bottom topographies 
for which there are extra conservation laws compared to the case of a generic bottom topography 
can provide new benchmarks against which these numerical models can be tested. 
We thus hope that the results reported in this paper could supply practitioners with information 
that is complementary to the exact solutions derived in seminal papers such as~\cite{carr58a,thac81a}.

Cases of extensions of the spaces of zeroth-order conservation laws
for systems from the class~\eqref{eq:2DSWEs} are classified 
up to the equivalence generated by the equivalence group~$G^\sim$ of this class. 
The consideration of this problem was started in~\cite{atam19a} 
but the results obtained there were not exhaustive. 
Earlier, a preliminary description of the above spaces was presented in~\cite{akse18b}.
We use an optimized version of the \emph{method of furcate splitting}, 
essentially following the proof of group classification of the class~\eqref{eq:2DSWEs} 
in~\cite{bihl19a}.
This is the first application of the method of furcate splitting to classifying conservation laws. 
Initially, the method was suggested in~\cite{niki01a} in the course of the group classification
of the class of nonlinear Schr\"odinger equations and was then applied to group classification 
of various classes of differential equations 
\cite{boyk01a,ivan10a,opan18b,vane15b,vane12a,vane15c,vasi01a}. 
Between some of the listed $G^\sim$-inequivalent cases of extensions of spaces of zeroth-order conservation laws, 
there are additional equivalences induced by admissible point transformations within the class~\eqref{eq:2DSWEs}
that are not generated jointly by elements of~$G^\sim$ and by point symmetry groups of systems from the class~\eqref{eq:2DSWEs}. 
Families of such admissible point transformations were found in~\cite{bihl19a}. 
We also extend the well-known Hamiltonian representation for the shallow water equations with flat bottom topography 
and the known family of distinguished (Casimir) functionals of the corresponding Hamiltonian operator
\cite[Eqs.~(3.6)--(3.7)]{salm88a}, see also \cite[Section~4.4]{shep90a}, 
to an arbitrary bottom topography, 
which allows us to relate the classifications of Lie symmetries and of zeroth-order conservation laws 
for systems from the class~\eqref{eq:2DSWEs}. 
To find generating sets of zeroth-order conservation laws under the action of Lie symmetries, 
we study this action and derive a convenient general formula 
describing the action of generalized symmetries on cosymmetries. 

\looseness=-1
The further organization of the paper is as follows.
For convenience of references, in Section~\ref{sec:TransPropertiesOfSWEs} we present 
results from~\cite{bihl19a} on the generalized equivalence group and the equivalence groupoid of the class~\eqref{eq:2DSWEs}
and on the classification of Lie symmetries of systems from this class. 
Section~\ref{sec:PreliminaryAnalysisOf0thOrderCLsOfSWEs} contains the preliminary analysis of determining equations for 
characteristics of zeroth-order conservation laws of systems from the class~\eqref{eq:2DSWEs} 
and the statement of classification results on such conservation laws.
The main points of the proof of the classification are presented in Section~\ref{sec:ClassificationProof} 
whereas details of applying the method of furcate splitting to it are moved to Section~\ref{sec:DetailsOfProof}. 
Hamiltonian structures of systems from the class~\eqref{eq:2DSWEs} are used in Section~\ref{sec:HamiltoniansOf2DSWEs} 
for interpreting and relating the classifications of Lie symmetries and of zeroth-order conservation laws. 
Therein we also compute generating sets (see~\cite{kham82a} for a definition) 
of such conservation laws under the action of Lie symmetries
for all classification cases that are inequivalent up to point transformations. 
In the final Section~\ref{sec:ConclusionsSWE} we summarize the findings of the paper and discuss possible future research directions.

\section[Equivalence group, Lie symmetries and admissible transformations]
{Equivalence group, Lie symmetries and admissible\\ transformations}\label{sec:TransPropertiesOfSWEs}

According to the formal definition of a class of differential equations \cite{bihl17a,kuru18a,popo06Ay,popo10Ay}, 
we call the class~\eqref{eq:2DSWEs} the set of systems of the form~\eqref{eq:2DSWEs}
with the independent variables $(t,x,y)$ and the dependent variables $(u,v,h)$. 
The parameter function~$b$ plays the role of the arbitrary element of this class 
and, in view of its definition, runs through the solution set of the auxiliary system 
\begin{gather*}
b_u=b_{u_t}=b_{u_x}=b_{u_y}=0, \\
b_v=b_{v_t}=b_{v_x}=b_{v_y}=0, \\
b_h=b_{h_t}=b_{h_x}=b_{h_y}=0, \\
b_t=0, 
\end{gather*}
which contains no inequalities for~$b$. 
For a fixed value of the arbitrary element~$b$, 
by~$\mathcal L_b$ we denote the system 
from the class~\eqref{eq:2DSWEs} with this value of~$b$.

By the general definition of generalized equivalence group \cite{bihl17a,kuru18a,popo06Ay,popo10Ay}, 
a generalized equivalence transformation of the class~\eqref{eq:2DSWEs} is a point transformation 
in the space with coordinates $(t,x,y,u_{(1)},v_{(1)},h_{(1)},b)$ 
that preserves the contact structure of the first-order jet space $\mathrm J^1(t,x,y|u,v,h)$
with the independent variables $(t,x,y)$ and the dependent variables $(u,v,h)$, 
which is coordinatized by $(t,x,y,u_{(1)},v_{(1)},h_{(1)})$. 
Here $u_{(1)}=(u_t,u_x,u_y)$, etc.
At the same time, the arbitrary element~$b$ does not depend on (first) derivatives of dependent variables,
and thus the space with the coordinates $(t,x,y,u,v,h,b)$
can be assumed as the space underlying the action 
of the generalized equivalence group~$G^\sim$ of the class~\eqref{eq:2DSWEs}; 
cf.\ the discussion in the beginning of Section~2 in~\cite{opan18b}.
Applying the algebraic method suggested in~\cite{hydo00By} and further developed in~\cite{bihl15a,bihl11Cy,card12Ay},
in~\cite{bihl19a} we proved the following theorem.

\begin{theorem}\label{2d_gen_equiv}
The generalized equivalence group~$G^\sim$ of the class~\eqref{eq:2DSWEs} 
coincides with the usual equivalence group of this class and consists of the transformations
\begin{gather*}
\tilde t=\delta_1t+\delta_2,\quad
\tilde x=\delta_3x-\varepsilon\delta_4y+\delta_5, \quad
\tilde y=\delta_4x+\varepsilon\delta_3y+\delta_6,\\[1ex]
\tilde u=\frac{\delta_3}{\delta_1}u-\varepsilon\frac{\delta_4}{\delta_1}v,\quad
\tilde v=\frac{\delta_4}{\delta_1}u+\varepsilon\frac{\delta_3}{\delta_1}v,\quad
\tilde h=\frac{\delta_3^{\,2}+\delta_4^{\,2}}{\delta_1^{\,2}}h,\quad
\tilde b=\frac{\delta_3^{\,2}+\delta_4^{\,2}}{\delta_1^{\,\,2}}b+\delta_7,
\end{gather*}
where $\varepsilon=\pm1$ and the parameters $\delta_i,\, i=1,\dots,7,$ are arbitrary constants with $\delta_1(\delta_3^{\ 2}+\delta_4^{\ 2})\ne0$.
\end{theorem}

The group~$G^\sim$ is generated by continuous equivalence transformations,
for which $\delta_1>0$ and $\varepsilon=1$, 
and two discrete equivalence transformations, 
which alternate signs of variables and are involutions,
\begin{gather*}
(t,x,y,u,v,h,b)\mapsto (-t,x,y,-u,-v,h,b), \\[.5ex]
(t,x,y,u,v,h,b)\mapsto (t,x,-y,u,-v,h,b).
\end{gather*}

Then we classified Lie symmetries of systems from the class~\eqref{eq:2DSWEs} up to the $G^\sim$-equivalence.%
\footnote{\looseness=1
The equivalence generated by the equivalence group of the class of systems of differential equations 
to be classified is a necessary element of the classical formulation of the group classification problem 
by Lie and Ovsiannikov~\cite[Chapter~III]{ovsi1982A} 
as well as of its natural extension to conservation laws~\cite[Section~3.3]{popo2008a}. 
The usage of this equivalence in the course of classifications 
of symmetry-like objects associated to systems of the class is justified by two mathematical facts. 
Firstly, a point transformation between two systems of differential equations 
induces an isomorphism between their maximal Lie invariance algebras 
(resp.\ between their spaces of local conservation laws, between their spaces of zeroth-order conservation laws, etc.). 
Secondly, the construction and the usage of the equivalence group of the class are in general much simpler 
than those of the corresponding equivalence groupoid. 
Moreover, the $G^\sim$-equivalence is also appropriate for the class~\eqref{eq:2DSWEs} 
from the physical point of view since the equivalence group~$G^\sim$ includes only simple transformations, 
which can be interpreted as physically relevant coordinate changes  
(shifts, scalings and reflections of the time and space variables and rotations of the space variables).
} 
For any fixed value of the arbitrary element~$b$, 
the maximal Lie invariance algebra~$\mathfrak g_b$ of the system~$\mathcal L_b$ 
consists of the vector fields of the form
\[2D(F^1)-c_1D^{\rm t}-c_2J+P(F^2,F^3),\]
where 
\begin{gather}\label{eq:2DSWEsRepresentationsForLieSymVFs}
\begin{split}
&D(F^1):=F^1\p_t+\tfrac12F^1_tx\p_x+\tfrac12F^1_ty\p_y-\tfrac12(F^1_tu-F^1_{tt}x)\p_u-\tfrac12(F^1_tv-F^1_{tt}y)\p_v-F^1_th\p_h,\\[1ex]
&D^{\rm t}:=t\p_t-u\p_u-v\p_v-2h\p_h,\quad
 J:=x\p_y-y\p_x+u\p_v-v\p_u,\\[1ex]
&P(F^2,F^3):=F^2\p_x+F^3\p_y+F^2_t\p_u+F^3_t\p_v,
\end{split}
\end{gather}
the parameters~$F^1$, $F^2$ and~$F^3$ are smooth functions of~$t$ 
and the parameters~$c_1$ and~$c_2$ are constants that satisfy the classifying equation
\begin{gather}\label{eq:2DSWEsClassifyingEq}
\begin{split}
&(F^1_tx+c_2y+F^2)b_x+(-c_2x+F^1_ty+F^3)b_y+2(F^1_t-c_1)b\\
&\qquad-F^1_{ttt}\frac{x^2+y^2}2-F^2_{tt}x-F^3_{tt}y-F^4=0,
\end{split}
\end{gather}
with $F^4$ being one more smooth parameter function of $t$. 
Up to antisymmetry, the nonzero commutation relations between the above vector fields 
are exhausted by the following ones:
\begin{gather}\label{eq:2DSWEsCommutationRelations}
\begin{split}
&[D(F^1),D(\tilde F^1)]=D(F^1\tilde F^1_t-\tilde F^1F^1_t),\quad
 [D(F^1),D^{\rm t}]=-D(tF^1_t-F^1),\\[1ex]
&[D(F^1),P(F^2,F^3)]=P\big(F^1F^2_t-\tfrac12F^1_tF^2,F^1F^3_t-\tfrac12F^1_tF^3\big),\\[1ex]
&[D^{\rm t},P(F^2,F^3)]=P(tF^2_t,tF^3_t),\quad 
 [J,P(F^2,F^3)]=P(F^3,-F^2).
\end{split}
\end{gather}

\newpage

The vector fields~\eqref{eq:2DSWEsRepresentationsForLieSymVFs} respectively generate 
the following local one-parameter groups of point transformations, 
where $\delta$ is the group parameter:
\begin{itemize}\itemsep=0ex
\item
$\tilde t=T,\ \tilde x=T_t^{1/2}x,\ \tilde y=T_t^{1/2}y,\ \tilde u=T_t^{-1/2}u+\frac12T_{tt}T_t^{-3/2}x,\ \tilde v=T_t^{-1/2}v+\frac12T_{tt}T_t^{-3/2}y$, $\tilde h=T_t^{-1}h$
with $T=T(t,\delta):=\hat H\big(H(t)+\delta\big)$, where $H$ is an antiderivative of~$1/F^1$, and $\hat H$~is the inverse of~$H$ with respect to~$t$, 
\ $\to$ \ arbitrary transformations of~$t$ with simultaneous linear transformations of the other variables with coefficients depending on~$t$, 
including shifts of~$t$ ($F^1=1$), concordant scalings of all variables ($F^1=t$) and time inversions ($F^1=t^2$).
\item
$\tilde t=e^\delta t,\ \tilde x=x,\ \tilde y=y,\ \tilde u=e^{-\delta}u,\ \tilde v=e^{-\delta}v,\ \tilde h=e^{-2\delta}h$ 
\ $\to$ \ scaling of~$t$ with simultaneous scalings of the dependent variables;
\item
$\tilde t=t,\ \tilde x=x\cos\delta-y\sin\delta,\ \tilde y=x\sin\delta+y\cos\delta,\ \tilde u=u\cos\delta-v\sin\delta,\ \tilde v=u\sin\delta+v\cos\delta,$ $\tilde h=h$ 
\ $\to$ \ simultaneous rotations in the $(x,y)$- and $(u,v)$-planes; 
\item
$\tilde t=t,\ \tilde x=x+\delta F^2(t),\ \tilde y=y+\delta F^3(t),\ \tilde u=u+\delta F^2_t(t),\ \tilde v=v+\delta F^3_t(t),\ \tilde h=h$ 
\ $\to$ \ generalized shifts of the space variables depending on~$t$, including usual shifts of the space variables ($F^2,F^3=\const$) and Galilean boosts ($F^2/t,F^3/t=\const$).
\end{itemize}

It is convenient to denote
\[D^{\rm s}:=2D(t)-2D^{\rm t}=x\p_x+y\p_y+u\p_u+v\p_v+2h\p_h\]
and to use hereafter, simultaneously with $(x,y)$, the polar coordinates~$(r,\varphi)$ on the $(x,y)$-plane,
\begin{gather*}
r:=\sqrt{x^2+y^2},\quad
\varphi:=\arctan\frac yx.
\end{gather*}

\begin{theorem}\label{thm:GroupClassificationOf2DSWEs1}
In the notation~\eqref{eq:2DSWEsRepresentationsForLieSymVFs}, 
the kernel Lie invariance algebra of the systems from the class~\eqref{eq:2DSWEs}
is $\mathfrak g^{\cap}=\langle D(1)\rangle$.
A~complete list of $G^\sim$-inequivalent Lie symmetry extensions within this class
is exhausted by the following cases,
where $f$ denotes an arbitrary smooth function of a single argument,
$\alpha$,~$\beta$, $\mu$ and~$\nu$ are arbitrary constants with $\alpha\geqslant0\bmod G^\sim$, $\beta>0$
and additional constraints indicated in the corresponding cases,
$\varepsilon=\pm1\bmod G^\sim$ and $\delta\in\{0,1\}\bmod G^\sim$.%
\footnote{%
The notation ``$\!{}\bmod G^\sim$'' means that the indicated constraints (here, $\varepsilon=\pm1$ and $\delta\in\{0,1\}$) 
can be set on the involved parameters ($\varepsilon$ and $\delta$, respectively) 
using transformations from the equivalence group~$G^\sim$.
}
\end{theorem}

\begin{enumerate}\itemsep=1.3ex
\item
\label{LieSymCase1} $b=r^\nu f(\varphi+\alpha\ln r)$, \ $(\alpha,\nu)\ne(0,-2)$, \ $\nu\ne0$:\quad
$\mathfrak g_b=\big\langle D(1),\,4D(t)-(\nu+2)D^{\rm t}-2\alpha J\big\rangle$;
\item
\label{LieSymCase2} $b=f(\varphi+\alpha\ln r)+\nu\ln r$, \ $\nu\in\{-1,0,1\}\bmod G^\sim$:\quad
$\mathfrak g_b=\big\langle D(1),\,2D(t)-D^{\rm t}-\alpha J\big\rangle$;
\item
\label{LieSymCase3} $b=f(r)+\delta\varphi$:\quad
$\mathfrak g_b=\big\langle D(1),\,J\big\rangle$;
\item
\label{LieSymCase4} $b=f(r)e^{\beta\varphi}$:\quad
$\mathfrak g_b=\big\langle D(1),\,2J-\beta D^{\rm t}\big\rangle$;
\item
\label{LieSymCase5} $ b=f(y)e^x$:\quad
$\mathfrak g_b=\big\langle D(1),\,D^{\rm t}-P(2,0)\big\rangle$;
\item
\label{LieSymCase6}
\begin{enumerate}
\item \label{LieSymCase6a}
$b=r^{-2}f(\varphi)$:\quad
$\mathfrak g_b=\big\langle D(1),\,D(t),\,D(t^2)\big\rangle$;
\item
\label{LieSymCase6b}
$b=r^{-2}f(\varphi)+\frac12r^2$:\quad
$\mathfrak g_b=\big\langle D(1),\,D(e^{2t}),\,D(e^{-2t})\big\rangle$;
\item
\label{LieSymCase6c}
$b=r^{-2}f(\varphi)-\frac12r^2$:\quad
$\mathfrak g_b=\big\langle D(1),\,D(\cos 2t),\,D(\sin 2t)\big\rangle$;
\end{enumerate}
\item
\label{LieSymCase7} $b=f(y)+\delta x$:\quad
$\mathfrak g_b=\big\langle D(1),\,P(1,0),\,P(t,0)\big\rangle$;
\item
\label{LieSymCase8} $b=f(y)+\tfrac12x^2$:\quad
$\mathfrak g_b=\big\langle D(1),\,P(e^t,0),\,P(e^{-t},0)\big\rangle$;
\item
\label{LieSymCase9} $b=f(y)-\tfrac12x^2$:\quad
$\mathfrak g_b=\big\langle D(1),\,P(\cos t,0),\,P(\sin t,0)\big\rangle$;
\item
\label{LieSymCase10} $b=\delta\varphi-\nu\ln r$, $\nu=\pm1\bmod G^\sim$ if $\delta=0$:\quad
$\mathfrak g_b=\big\langle D(1),\,2D(t)-D^{\rm t},\,J\big\rangle$;
\item
\label{LieSymCase11} $b=\varepsilon r^{\nu}e^{\alpha\varphi}$,
$\nu\ne-2$, $(\alpha,\nu)\notin\{(0,0),(0,2)\}$:\ \
$\mathfrak g_b=\big\langle D(1),\,4D(t)-(\nu+2)D^{\rm t},\,2J-\alpha D^{\rm t}\big\rangle$;
\item
\label{LieSymCase12}
\begin{enumerate}
\item
\label{LieSymCase12a} $b=\varepsilon r^{-2}e^{\alpha\varphi}$:\quad
$\mathfrak g_b=\big\langle D(1),\,D(t),\,D(t^2),\,\alpha D^{\rm s}+4J\big\rangle$;
\item
\label{LieSymCase12b} $b=\varepsilon r^{-2}e^{\alpha\varphi}+\frac12r^2$:\quad
$\mathfrak g_b=\big\langle D(1),\,D(e^{2t}),\,D(e^{-2t}),\,\alpha D^{\rm s}+4J\big\rangle$;
\item
\label{LieSymCase12c} $b=\varepsilon r^{-2}e^{\alpha\varphi}-\frac12r^2$:\quad
$\mathfrak g_b=\big\langle D(1),\,D(\cos 2t),\,D(\sin 2t),\,\alpha D^{\rm s}+4J\big\rangle$;
\end{enumerate}
\item
\label{LieSymCase13} $b=\varepsilon|y|^{\nu}+\delta x$, $\nu\notin\{-2,0,2\}$:\\
$\mathfrak g_b=\big\langle D(1),\,4D(t)-(\nu+2)D^{\rm t}-\delta(\nu-1)P(t^2,0),\,P(1,0),\,P(t,0)\big\rangle$;
\item
\label{LieSymCase14} $b=\varepsilon\ln|y|+\delta x$:\quad
$\mathfrak g_b=\big\langle D(1),\,2D(t)-D^{\rm t}-\frac12\delta P(t^2,0)),\,P(1,0),\,P(t,0)\big\rangle;$
\item
\label{LieSymCase15} $b=\varepsilon e^y+\delta x$:\quad
$\mathfrak g_b=\big\langle D(1),\,D^{\rm t}-P(\delta t^2,2),\,P(1,0),\,P(t,0)\big\rangle;$
\item\label{LieSymCase16}
\begin{enumerate}
\item
\label{LieSymCase16a} $b=\varepsilon y^{-2}+\delta x$:\quad
$\mathfrak g_b=\big\langle D(1),\,D(t)+\frac34\delta P(t^2,0),\,D(t^2)+\frac12\delta P(t^3,0),\,P(1,0),\,P(t,0)\big\rangle;$
\item
\label{LieSymCase16b} $b=\varepsilon y^{-2}+\frac12r^2$:\quad
$\mathfrak g_b=\big\langle D(1),\,D(e^{2t}),\,D(e^{-2t}),\,P(e^t,0),\,P(e^{-t},0)\big\rangle$;
\item
\label{LieSymCase16c} $b=\varepsilon y^{-2}-\frac12r^2$:\quad
$\mathfrak g_b=\big\langle D(1),\,D(\cos 2t),\,D(\sin 2t),\,P(\cos t,0),\,P(\sin t,0)\big\rangle$;
\end{enumerate}
\item
\label{LieSymCase17} $b=\tfrac12x^2+\tfrac12\beta^2y^2$, $0<\beta<1$:\quad
$\mathfrak g_b=\big\langle D(1),\,D^{\rm s},\,P(e^t,0),\,P(e^{-t},0),\,P(0,e^{\beta t}),\,P(0,e^{-\beta t})\big\rangle$;
\item
\label{LieSymCase18} $b=\frac12x^2+\delta y$:\quad
$\mathfrak g_b=\big\langle D(1),\,D^{\rm s}-\tfrac12\delta P(0,t^2),\,P(e^t,0),\,P(e^{-t},0),\,P(0,1),\,P(0,t)\big\rangle$;
\item
\label{LieSymCase19} $b=\tfrac12x^2-\tfrac12\beta^2y^2$, $\beta>0$:\quad
$\mathfrak g_b=\big\langle D(1),\,D^{\rm s},\,P(e^t,0),\,P(e^{-t},0),\,P(0,\cos\beta t),\,P(0,\sin\beta t)\big\rangle$;
\item
\label{LieSymCase20} $b=-\frac12x^2+\delta y$:\quad
$\mathfrak g_b=\big\langle D(1),\,D^{\rm s}-\tfrac12\delta P(0,t^2),\,P(\cos t,0),\,P(\sin t,0),\,P(0,1),\,P(0,t)\big\rangle$;
\item
\label{LieSymCase21} $b=-\tfrac12x^2-\tfrac12\beta^2y^2$, $0<\beta<1$:\\[2pt]
$\mathfrak g_b=\big\langle D(1),\,D^{\rm s},\,P(\cos t,0),\,P(\sin t,0),\,P(0,\cos\beta t),\,P(0,\sin\beta t)\big\rangle$;
\item
\label{LieSymCase22}
\begin{enumerate}
\item
\label{LieSymCase22a} $b=0$:\quad
$\mathfrak g_b=\big\langle D(1),\,D(t),\,D(t^2),\,D^{\rm s},\,J,\,P(1,0),\,P(t,0),\,P(0,1),\,P(0,t)\big\rangle$;
\item
\label{LieSymCase22b} $b=x$:\quad
$\mathfrak g_b=\big\langle D(1),\,D(t)+\tfrac34P(t^2,0),\,D(t^2)+\tfrac12P(t^3,0),\,D^{\rm s}-\tfrac12P(t^2,0),\,J-\tfrac12P(0,t^2)$,\\[.5ex]
$\phantom{\mbox{$b=x$:\quad}\mathfrak g_b=\big\langle}P(1,0),\,P(t,0),\,P(0,1),\,P(0,t)\big\rangle$;
\item
\label{LieSymCase22c} $b=\frac12r^2$:\ \
$\mathfrak g_b=\big\langle D(1),\,D(e^{2t}),\,D(e^{-2t}),\,D^{\rm s},\,J,\,P(e^t,0),\,P(e^{-t},0),\,P(0,e^t),\,P(0,e^{-t})\big\rangle$;\!
\item
\label{LieSymCase22d} $b=-\frac12r^2$:\quad
$\mathfrak g_b=\big\langle D(1),\,D(\cos2t),\,D(\sin2t),\,D^{\rm s},\,J$,\\[.5ex]
$\phantom{\mbox{$b=-\frac12r^2$:\quad}\mathfrak g_b=\big\langle}P(\cos t,0),\,P(\sin t,0),\,P(0,\cos t),\,P(0,\sin t)\big\rangle$.
\end{enumerate}
\end{enumerate}

In~\cite{bihl19a}, we showed that the class~\eqref{eq:2DSWEs} is not semi-normalized;
see \cite{bihl11Dy,kuru18a,opan17a,popo06Ay,popo10Ay} for definitions.
Analyzing the structure of the maximal Lie invariance algebras of systems from the class~\eqref{eq:2DSWEs}, 
we found three families of 
$G^\sim$-inequivalent non-identity admissible transformations of the class~\eqref{eq:2DSWEs}
that are independent up to inversion, composing with each other 
and with admissible transformations generated by point symmetries of systems from this class,
{\renewcommand{\labelenumi}{$\mathcal T^\theenumi$:}
\begin{enumerate}\itemsep=1.5ex
\item
$b=r^{-2}f(\varphi)-\frac12r^2$, \ $\tilde b=\tilde r^{-2}f(\tilde\varphi)$,\\[1ex]
$\tilde t=\tan t$, \ $\tilde x=\dfrac x{\cos t}$, \ $\tilde y=\dfrac y{\cos t}$, \ $\tilde u=u\cos t+x\sin t$, \ $\tilde v=v\cos t+y\sin t$, \ $\tilde h=h\cos^2 t$;
\item
$b=r^{-2}f(\varphi)+\frac12r^2$, \ $\tilde b=\tilde r^{-2}f(\tilde\varphi)$,\\[1ex]
$\tilde t=\tfrac12e^{2t}$, \ $\tilde x=e^tx$, \ $\tilde y=e^ty$, \ $\tilde u=e^{-t}\left( u+x\right)$, \ $\tilde v=e^{-t}\left( v+y\right)$, \ $\tilde h=e^{-2t}h$;
\item
$b=f(y)+x$, \ $\tilde b=f(\tilde y)$,\\[1ex]
$\tilde t=t$, \ $\tilde x=x+\tfrac12t^2$, \ $\tilde y=y$, \ $\tilde u=u+t$, \ $\tilde v=v$, \ $\tilde h=h$.
\end{enumerate}
These families of admissible transformations induce additional equivalences among classification cases of Theorem~\ref{thm:GroupClassificationOf2DSWEs1},
\begin{enumerate}\itemsep=0ex\itemindent=2em
\item
\ref{LieSymCase6c}  $\to$ \ref{LieSymCase6a},  \
\ref{LieSymCase12c} $\to$ \ref{LieSymCase12a}, \
\ref{LieSymCase16b} $\to$ \ref{LieSymCase16a}$_{\delta=0}$, \
\ref{LieSymCase22d} $\to$ \ref{LieSymCase22a};
\item
\ref{LieSymCase6b}  $\to$ \ref{LieSymCase6a},  \
\ref{LieSymCase12b} $\to$ \ref{LieSymCase12a}, \
\ref{LieSymCase16c} $\to$ \ref{LieSymCase16a}$_{\delta=0}$, \
\ref{LieSymCase22c} $\to$ \ref{LieSymCase22a};
\item
\ref{LieSymCase22b} $\to$ \ref{LieSymCase22a}, \quad
\ref{LieSymCase7}, \ref{LieSymCase13}, \ref{LieSymCase14}, \ref{LieSymCase15}, \ref{LieSymCase16a}, \ref{LieSymCase18}, \ref{LieSymCase20}:%
\footnote{\label{fnt:OnModificationOfAdmTrans}%
In Cases~\ref{LieSymCase18} and~\ref{LieSymCase20}, instead of the pure transformational part of $\mathcal T^3$
one should use its composition with
the permutation $(x,u)\leftrightarrow(y,v)$, 
which is associated with a discrete equivalence transformations of the class~\eqref{eq:2DSWEs}.
}
\ $\delta=1$ $\to$ $\delta=0$.
\end{enumerate}
}

\begin{remark}\label{rem:Two-LevelNumerationOfClassificationCases}
The presence of additional equivalences among classification cases listed in Theorem~\ref{thm:GroupClassificationOf2DSWEs1}
justifies the usage of two-level numeration for them: 
numbers with the same Arabic numerals and different Roman letters
correspond to cases that are equivalent with respect to additional equivalence transformations. 
\end{remark}

\begin{remark}\label{rem:LieSymsComparisonWithEulerEqs}
Lie-symmetry properties of the shallow water equations differ from those of the (incompressible) Euler equations. 
Whereas the maximal Lie invariance algebra of the Euler equations is infinite-dimensional%
\footnote{%
The maximal Lie invariance algebra~$\mathfrak g_{\rm EEs}^{\rm 3D}$ 
of the Euler equations $\mathbf u_t+(\mathbf u\cdot\nabla)\mathbf u+\nabla p=0$, $\nabla\cdot\mathbf u=0$
in space dimension three 
was first computed in~\cite{buch1971a}, and this computation was enhanced in~\cite{khab76a}.
See also \cite[Example~2.45]{olve93a} and \cite[Section~1.1.11]{andr98A} for accessible presentations of these results. 
For the standard representation of the two-dimensional Euler equations in terms of the velocity~$\mathbf u$ and the pressure~$p$, 
the maximal Lie invariance algebra is completely analogous to~$\mathfrak g_{\rm EEs}^{\rm 3D}$, 
whereas the representation with excluded pressure additionally admits, as its Lie symmetries, 
the time-dependent rotations with constant angular velocities~\cite[Section~1.2.2]{andr98A}. 
These additional Lie symmetries are not relevant for the shallow water equations of the form~\eqref{eq:2DSWEs}.
Note that the maximal Lie invariance algebra of any variant of the (incompressible) Navier--Stokes equations 
is a codimension-one subalgebra of its counterpart for zero viscosity coefficient, 
where two independent vector fields associated with scaling transformations 
are linearly combined into a single one, 
see~\cite{byte1972a} for space dimension three. 
} 
the algebra~$\mathfrak g_b$ is finite-dimensional for any bottom topography~$b$. 
The Euler equations admit generalized shifts of the space variables that arbitrarily depends on~$t$, 
and systems from the class~\eqref{eq:2DSWEs} are invariant 
with respect to at most four independent generalized shifts of the space variables 
with a specific (affine, exponential or trigonometric) dependence on~$t$.
At the same time, some systems from the class~\eqref{eq:2DSWEs} admits point symmetry transformations 
with more complicated $t$-components than those for the Euler equations, 
which are merely affine~in~$t$.
\end{remark}

\section{Preliminary analysis and classification result}\label{result_section}\label{sec:PreliminaryAnalysisOf0thOrderCLsOfSWEs}

Suppose that for a fixed value of the arbitrary element~$b$, 
a tuple~$\gamma=(\gamma^1,\gamma^2,\gamma^3)^{\mathsf T}$ of sufficiently smooth functions of $(t,x,y,u,v,h)$
is a characteristic of a conservation law of the system~$\mathcal L_b$.
A~conservation law of~$\mathcal L_b$ is of order zero if and only if it admits a characteristic of this kind. 
By~${\rm Ch}^0_b$ we denote the linear space constituted by such tuples. 
Then the tuple~$\gamma$ is a cosymmetry of~$\mathcal L_b$, 
which is equivalent to satisfying the condition \cite[Eq.~(5.83)]{olve93a} 
\begin{gather}\label{eq:CosymCondition}
(\mathsf D_{\mathcal L_b}^\dag\gamma)\,\big|_{\mathcal L_b}=0,
\end{gather}
where $\mathsf D_{\mathcal L_b}^\dag$ is the adjoint of the Fr\'echet derivative of the left-hand side of the system~$\mathcal L_b$, 
\begin{gather*}
\mathsf D_{\mathcal L_b}^\dag=-
\begin{pmatrix}
\mathrm D_t+u\mathrm D_x+v\mathrm D_y+v_y & -v_x & h\mathrm D_x\\
-u_y & \mathrm D_t+u\mathrm D_x+v\mathrm D_y+u_x & h\mathrm D_y\\
\mathrm D_x & \mathrm D_y & \mathrm D_t+u\mathrm D_x+v\mathrm D_y
\end{pmatrix},
\end{gather*}
with $\mathrm D_t$, $\mathrm D_x$ and $\mathrm D_y$ denoting the total derivative operators with respect to~$t$, $x$ and~$y$, 
respectively. 
Since $\mathcal L_b$ is a system of evolution equations, 
it is natural to assume the derivatives $(u_t,v_t,h_t)$ as the leading ones 
and express them in terms of other (parametric) derivatives in view of the system~$\mathcal L_b$, 
\begin{gather*}
u_t=-uu_x-vu_y-h_x+b_x,\\
v_t=-uv_x-vv_y-h_y+b_y,\\
h_t=-(uh)_x-(vh)_y.
\end{gather*}
To get the system of determining equations for the components of the tuple~$\gamma$,
we first confine the expanded expression for $\mathsf D_{\mathcal L_b}^\dag\gamma$ in the condition~\eqref{eq:CosymCondition}
to the manifold defined by $\mathcal L_b$ in the jet space $\mathrm J^1(t,x,y|u,v,h)$
via substituting the above expressions for the leading derivatives.
The derived equations are then split with respect to 
the first-order (parametric) derivatives of the dependent variables $(u,v,h)$ with respect to~$x$ and~$y$.
We additionally rearrange the equations obtained after the splitting 
and exclude those of them that are differential consequences of the others,
which finally gives the system 
\begin{subequations}\label{eq:DetEqsForChars}
\begin{gather}
\label{eq:DetEqsForChars1}
\gamma^1_v=\gamma^2_u=0,\quad \gamma^1_u=\gamma^2_v=h\gamma^3_h,\quad \gamma^1_h=\gamma^3_u,\quad \gamma^2_h=\gamma^3_v,\quad h\gamma^1_h=\gamma^1,\quad h\gamma^2_h=\gamma^2,\\
\label{eq:DetEqsForChars2}
\gamma^1_t+u\gamma^1_x+v\gamma^1_y+b_x\gamma^1_u+h\gamma^3_x=0,\\
\label{eq:DetEqsForChars3}
\gamma^2_t+u\gamma^2_x+v\gamma^2_y+b_y\gamma^2_v+h\gamma^3_y=0,\\
\label{eq:DetEqsForChars4}
\gamma^3_t+u\gamma^3_x+v\gamma^3_y+b_x\gamma^3_u+b_y\gamma^3_v+\gamma^1_x+\gamma^2_y=0.
\end{gather}
\end{subequations}
The equations $\gamma^1_v=\gamma^2_u$, $\gamma^1_h=\gamma^3_u$ and $\gamma^2_h=\gamma^3_v$ imply 
that the Fr\'echet derivative of~$\gamma$ is a formally self-adjoint operator, 
which is a necessary and sufficient condition for a cosymmetry of an evolution system 
to be a conservation-law characteristic of this system.
In other words, the space of nonpositive-order cosymmetries of~$\mathcal L_b$ coincides with~${\rm Ch}^0_b$. 

The general solution of the subsystem~\eqref{eq:DetEqsForChars1} can be represented in the form 
\begin{gather*}
\gamma^1=-2F^1uh+\gamma^{10}h, \\
\gamma^2=-2F^1vh+\gamma^{20}h, \\
\gamma^3=-F^1(u^2+v^2+2h)+\gamma^{10}u+\gamma^{20}v+\gamma^{30},
\end{gather*}
where the coefficients $F^1$, $\gamma^{10}$, $\gamma^{20}$ and~$\gamma^{30}$ depend at most on $(t,x,y)$.
We substitute this representation into the equations~\eqref{eq:DetEqsForChars2}--\eqref{eq:DetEqsForChars4} 
and split the resulting equations with respect to~$(u,v,h)$ to derive a system for representation's coefficients, 
\begin{subequations}\label{eq:DetEqsForCharCoeffs}
\begin{gather}
\label{eq:DetEqsForCharCoeffs1}
F^1_x=F^1_y=0,\quad \gamma^{10}_x=\gamma^{20}_y=F^1_t,\quad \gamma^{10}_y+\gamma^{20}_x=0,\\
\label{eq:DetEqsForCharCoeffs2}
\gamma^{30}_x=2F^1b_x-\gamma^{10}_t,\quad 
\gamma^{30}_y=2F^1b_y-\gamma^{20}_t,\\
\label{eq:DetEqsForCharCoeffs3}
\gamma^{10}b_x+\gamma^{20}b_y+\gamma^{30}_t=0.
\end{gather}
\end{subequations}
Differential consequences of~\eqref{eq:DetEqsForCharCoeffs1}--\eqref{eq:DetEqsForCharCoeffs2} are 
the equations $\gamma^{10}_{tx}=\gamma^{20}_{ty}=0$ and $\gamma^{i0}_{xx}=\gamma^{i0}_{xy}=\gamma^{i0}_{yy}=0$, $i=1,2$.
Thus, the subsystem~\eqref{eq:DetEqsForCharCoeffs1}--\eqref{eq:DetEqsForCharCoeffs2} 
integrates to 
\begin{gather*}
\gamma^{10}=F^1_tx+c_1y+F^2,\\ 
\gamma^{20}=-c_1x+F^1_ty+F^3,\\
\gamma^{30}=2F^1b-F^1_{tt}\frac{x^2+y^2}{2}-F^2_tx-F^3_ty+F^4, 
\end{gather*}
where $F^i$, $i=1,2,3,4$, are sufficiently smooth functions of $t$, and $c_1$ is a constant. 
This results in a refined form of the components of the tuple~$\gamma$, 
\begin{gather}\label{eq:FormOfChars}
\begin{split}
\gamma^1={}&(-2F^1u+F^1_tx+c_1y+F^2)h,\\
\gamma^2={}&(-2F^1v-c_1x+F^1_ty+F^3)h,\\
\gamma^3={}&-F^1(u^2+v^2+2h)+(F^1_tx+c_1y+F^2)u+(-c_1x+F^1_ty+F^3)v\\
&+2F^1b-F^1_{tt}\frac{x^2+y^2}2-F^2_tx-F^3_ty+F^4,
\end{split}
\end{gather}
where, as above, $F^i$, $i=1,2,3,4$, are sufficiently smooth functions of $t$, and $c_1$ is a constant.
In~other words, for any~$b$
\[
{\rm Ch}^0_b\subset{\rm Ch}^0_\spanindex:=\langle \Lambda^0,\,\Lambda^1_b(F^1),\,\Lambda^2(F^2),\,\Lambda^3(F^3),\,\Lambda^4(F^4)\rangle,
\]
where the parameters~$F^1$, $F^2$, $F^3$ and~$F^4$ run through the set of smooth functions of~$t$,
\begin{gather}\label{eq:SpanningCanditatesForChars}
\begin{split}
&\Lambda^0:=\big(-yh,\,xh,\,xv-yu\big),\\
&\Lambda^1_b(F^1):=\big(-2F^1uh+F^1_txh,\,-2F^1vh+F^1_tyh,\\
&\phantom{\Lambda^1_b(F^1):=\big(}-F^1(u^2+v^2+2h)+F^1_t(xu+yv)+2F^1b-\tfrac12F^1_{tt}(x^2+y^2)\big),\\
&\Lambda^2(F^2):=\big(F^2h,\,0,\,F^2u-F^2_tx\big),\\
&\Lambda^3(F^3):=\big(0,\,F^3h,\,F^3v-F^3_ty\big),\\
&\Lambda^4(F^4):=\big(0,\,0,\,F^4\big).
\end{split}
\end{gather}

For elements of~${\rm Ch}^0_b$, 
the parameters~$F^1$, $F^2$, $F^3$, $F^4$ and~$c_1$ additionally satisfy the equation implied by~\eqref{eq:DetEqsForCharCoeffs3},  
\begin{gather}\label{eq:2DSWEsClassifyingEqFor0thOrderCLs}
\begin{split}
&(F^1_tx+c_1y+F^2)b_x+(-c_1x+F^1_ty+F^3)b_y+2F^1_tb\\
&\qquad-F^1_{ttt}\frac{x^2+y^2}{2}-F^2_{tt}x-F^3_{tt}y-F^4_t=0.
\end{split}
\end{gather}
The equation~\eqref{eq:2DSWEsClassifyingEqFor0thOrderCLs} is the only classifying equation for 
nonpositive-order conservation-law characteristics of systems from the class~\eqref{eq:2DSWEs}.
Thus, the problem of classification of such characteristics up to the $G^\sim$-equivalence
reduces to solving the equation~\eqref{eq:2DSWEsClassifyingEqFor0thOrderCLs} up to the same equivalence
with respect to the arbitrary element $b$ and the parameters~$F^1$, \dots, $F^4$ and~$c_1$.

A zeroth-order conserved current~$\mathcal C=(\mathcal C^1,\mathcal C^2,\mathcal C^3)$ associated with the characteristic 
\[\gamma=c_1\Lambda^0+\Lambda^1_b(F^1)+\Lambda^2(F^2)+\Lambda^3(F^3)+\Lambda^4(F^4)\] 
can be computed directly from 
the characteristic form of the condition for conserved currents, 
\begin{gather*}
\mathrm D_t\mathcal C^1+\mathrm D_x\mathcal C^2+\mathrm D_y\mathcal C^3
=\gamma^1(u_t+uu_x+vu_y+h_x-b_x)+\gamma^2(v_t+uv_x+vv_y+h_y-b_y)\\
\qquad{}+\gamma^3(h_t+(uh)_x+(vh)_y).
\end{gather*}
``Integrating by parts'' in the right-hand side of the last equality 
and taking into account the classifying equation~\eqref{eq:2DSWEsClassifyingEqFor0thOrderCLs}, 
we obtain that up to adding a trivial conserved current 
\begin{gather}\label{eq:RepresentationOfGenCCs}
\mathcal C=c_1\mathcal F^0+\mathcal F^1_b(F^1)+\mathcal F^2(F^3)+\mathcal F^3(F^3)+\mathcal F^4(F^4)
\end{gather}
with
\begin{gather*}
\begin{split}
\mathcal F^0
:={}&\big((xv-yu)h,\,(xv-yu)uh-\tfrac12yh^2,\,(xv-yu)uh+\tfrac12xh^2\big),\\ 
 ={}&(xv-yu)\big(1,u,v\big)+\tfrac12h^2(0,-y,x),
\end{split}
\\[1ex]
\begin{split}
\mathcal F^1_b(F^1)
:={}&\big(
-F^1(u^2+v^2+h-2b)h+F^1_t(xu+yv)h-\tfrac12F^1_{tt}(x^2+y^2)h,\,\\&\phantom{\big(}
-F^1(u^2+v^2+2h-2b)uh+F^1_t((xu+yv)u+\tfrac12xh)h-\tfrac12F^1_{tt}(x^2+y^2)uh,\,\\&\phantom{\big(}
-F^1(u^2+v^2+2h-2b)vh+F^1_t((xu+yv)v+\tfrac12yh)h-\tfrac12F^1_{tt}(x^2+y^2)vh\big),\\
 ={}&\big(-F^1(u^2+v^2+2h-2b)+F^1_t(xu+yv)h-\tfrac12F^1_{tt}(x^2+y^2)\big)\big(1,u,v\big)\\
    &+F^1h^2(1,0,0)+\tfrac12F^1_th^2(0,x,y),
\end{split}
\\[1ex]
\begin{split}
\mathcal F^2(F^2)
:={}&\big((F^2u-F^2_tx)h,\,(F^2u-F^2_tx)uh+\tfrac12F^2h^2,\,(F^2u-F^2_tx)vh\big),\\ 
 ={}&(F^2u-F^2_tx)h\big(1,u,v\big)+\tfrac12F^2h^2(0,1,0),
\end{split}
\\[1ex]
\begin{split}
\mathcal F^3(F^3)
:={}&\big((F^3v-F^3_ty)h,\,(F^3v-F^3_ty)uh,\,(F^3v-F^3_ty)vh+\tfrac12F^3h^2\big)\\
 ={}&(F^3v-F^3_ty)h\big(1,u,v\big)+\tfrac12F^3h^2(0,0,1),
\end{split}
\\[1ex]
\mathcal F^4(F^4):=\big(F^4h,\,F^4uh,\,F^4vh\big)=F^4h\big(1,u,v\big). 
\end{gather*}

Recall that $r:=\sqrt{x^2+y^2}$ and $\varphi:=\arctan(y/x)$ are 
the polar coordinates~$(r,\varphi)$ on the $(x,y)$-plane.

\begin{theorem}\label{thm:ClassificationOf0thOrderCLsOf2DSWEs1}
In the notation~\eqref{eq:SpanningCanditatesForChars}, for any value of the arbitrary element~$b$, 
the system~$\mathcal L_b$ admits two linearly independent conservation laws with characteristics~$\Lambda^1_b(1)$ and~$\Lambda^4(1)$, 
${\rm Ch}^0_b\supseteq\langle\Lambda^4(1),\,\Lambda^1_b(1)\rangle$.
A~complete list of $G^\sim$-inequivalent extensions of the spaces of nonpositive-order conservation-law characteristics,
${\rm Ch}^0_b$, within the class~\eqref{eq:2DSWEs}
is exhausted by the following cases,
where $f$ denotes an arbitrary smooth function of a single argument,
$\alpha$,~$\beta$, $\mu$ and~$\nu$ are arbitrary constants with $\alpha\geqslant0\bmod G^\sim$, $\beta>0$
and additional constraints indicated in the corresponding cases,
$\varepsilon=\pm1\bmod G^\sim$ and $\delta\in\{0,1\}\bmod G^\sim$.
\end{theorem}

\begin{enumerate}\itemsep=1.5ex
\item
\label{CLsCase1} $b=r^{-2}f(\varphi+\beta\ln r)$:\quad
${\rm Ch}^0_b=\big\langle\Lambda^4(1),\,\Lambda^1_b(1),\,\Lambda^1_b(t)+\beta\Lambda^0\big\rangle$;
\item
\label{CLsCase2} $b=f(r)+\delta\varphi:$\quad
${\rm Ch}^0_b=\big\langle\Lambda^4(1),\,\Lambda^1_b(1),\,\Lambda^0+\delta\Lambda^4(t)\big\rangle$;
\item
\label{CLsCase3}
\begin{enumerate}\itemsep=.5ex
\item 
\label{CLsCase3a} $b=r^{-2}f(\varphi)$:\quad
${\rm Ch}^0_b=\big\langle\Lambda^4(1),\,\Lambda^1_b(1),\,\Lambda^1_b(t),\,\Lambda^1_b(t^2)\big\rangle$;
\item
\label{CLsCase3b} $b=r^{-2}f(\varphi)+\frac12r^2$:\quad
${\rm Ch}^0_b=\big\langle\Lambda^4(1),\,\Lambda^1_b(1),\,\Lambda^1_b(e^{2t}),\,\Lambda^1_b(e^{-2t})\big\rangle$;
\item
\label{CLsCase3c} $b=r^{-2}f(\varphi)-\frac12r^2$:\quad
${\rm Ch}^0_b=\big\langle\Lambda^4(1),\,\Lambda^1_b(1),\,\Lambda^1_b(\cos 2t),\,\Lambda^1_b(\sin 2t)\big\rangle$;
\end{enumerate}
\item
\label{CLsCase4} $b=f(y)+\delta x$:\quad
${\rm Ch}^0_b=\big\langle\Lambda^4(1),\,\Lambda^1_b(1),\,\Lambda^2(1),\,\Lambda^2(t)\big\rangle$;
\item
\label{CLsCase5} $b=f(y)+\tfrac12x^2$:\quad
${\rm Ch}^0_b=\big\langle\Lambda^4(1),\,\Lambda^1_b(1),\,\Lambda^2(e^t),\,\Lambda^2(e^{-t})\big\rangle$;
\item
\label{CLsCase6}$b=f(y)-\tfrac12x^2$:\quad
${\rm Ch}^0_b=\big\langle\Lambda^4(1),\,\Lambda^1_b(1),\,\Lambda^2(\cos t),\,\Lambda^2(\sin t)\big\rangle$;
\item
\label{CLsCase7} 
\begin{enumerate}\itemsep=.5ex
\item
\label{CLsCase7a} $b=\varepsilon r^{-2}$:\quad
${\rm Ch}^0_b=\big\langle\Lambda^4(1),\,\Lambda^1_b(1),\,\Lambda^1_b(t),\,\Lambda^1_b(t^2),\,\Lambda^0\big\rangle$;
\item
\label{CLsCase7b} $b=\varepsilon r^{-2}+\frac12r^2$:\quad
${\rm Ch}^0_b=\big\langle\Lambda^4(1),\,\Lambda^1_b(1),\,\Lambda^1_b(e^{2t}),\,\Lambda^1_b(e^{-2t}),\,\Lambda^0\big\rangle$;
\item
\label{CLsCase7c} $b=\varepsilon r^{-2}-\frac12r^2$:\quad
${\rm Ch}^0_b=\big\langle\Lambda^4(1),\,\Lambda^1_b(1),\,\Lambda^1_b(\cos 2t),\,\Lambda^1_b(\sin 2t),\,\Lambda^0\big\rangle$;
\end{enumerate}
\item
\label{CLsCase8}
\begin{enumerate}\itemsep=.5ex
\item
\label{CLsCase8a} $b=\varepsilon y^{-2}+\delta x$:\quad
${\rm Ch}^0_b=\big\langle\Lambda^4(1),\,\Lambda^1_b(1),\,
\Lambda^1_b(t)+\frac32\delta\Lambda^2(t^2)+\frac12\delta^2\Lambda^4(t^3)$,\\[.5ex]
${}$\qquad\qquad$\Lambda^1_b(t^2)+\delta\Lambda^2(t^3)+\frac14\delta^2\Lambda^4(t^4),\,\Lambda^2(1)+\delta\Lambda^4(t),\,\Lambda^2(t)+\frac12\delta\Lambda^4(t^2)\big\rangle;$
\item
\label{CLsCase8b} $b=\varepsilon y^{-2}+\frac12r^2$:\quad
${\rm Ch}^0_b=\big\langle\Lambda^4(1),\,\Lambda^1_b(1),\,\Lambda^1_b(e^{2t}),\,\Lambda^1_b(e^{-2t}),\,\Lambda^2(e^t),\,\Lambda^2(e^{-t})\big\rangle$;
\item
\label{CLsCase8c} $b=\varepsilon y^{-2}-\frac12r^2$:\quad
${\rm Ch}^0_b=\big\langle\Lambda^4(1),\,\Lambda^1_b(1),\,\Lambda^1_b(\cos 2t),\,\Lambda^1_b(\sin 2t),\,\Lambda^2(\cos t),\,\Lambda^2(\sin t)\big\rangle$;
\end{enumerate}
\item
\label{CLsCase9} $b=\tfrac12x^2+\tfrac12\beta^2y^2$, \ $0<\beta<1$:\quad
${\rm Ch}^0_b=\big\langle\Lambda^4(1),\,\Lambda^1_b(1),\,\Lambda^2(e^t),\,\Lambda^2(e^{-t}),\,\Lambda^3(e^{\beta t}),\,\Lambda^3(e^{-\beta t})\big\rangle$;
\item
\label{CLsCase10} $b=\frac12x^2+\delta y$:\quad
${\rm Ch}^0_b=\big\langle\Lambda^4(1),\,\Lambda^1_b(1),\,\Lambda^2(e^t),\,\Lambda^2(e^{-t}),\,\Lambda^3(1)+\delta\Lambda^4(t),\,\Lambda^3(t)+\frac12\delta\Lambda^4(t^2)\big\rangle$;
\item
\label{CLsCase11} $b=\tfrac12x^2-\tfrac12\beta^2y^2$, \ $\beta>0$:\quad
${\rm Ch}^0_b=\big\langle\Lambda^4(1),\,\Lambda^1_b(1),\,\Lambda^2(e^t),\,\Lambda^2(e^{-t}),\,\Lambda^3(\cos\beta t),\,\Lambda^3(\sin\beta t)\big\rangle$;
\item
\label{CLsCase12} $b=-\frac12x^2+\delta y$: \ 
${\rm Ch}^0_b=\big\langle\Lambda^4(1),\Lambda^1_b(1),\Lambda^2(\cos t),\Lambda^2(\sin t),\Lambda^3(1)+\delta\Lambda^4(t),\Lambda^3(t)+\frac12\delta\Lambda^4(t^2)\big\rangle$;
\item
\label{CLsCase13} $b=-\tfrac12x^2-\tfrac12\beta^2y^2$, \ $0<\beta<1$:\\[2pt]
${\rm Ch}^0_b=\big\langle\Lambda^4(1),\,\Lambda^1_b(1),\,\Lambda^2(\cos t),\,\Lambda^2(\sin t),\,\Lambda^3(\cos\beta t),\,\Lambda^3(\sin\beta t)\big\rangle$;
\item
\label{CLsCase14}
\begin{enumerate}\itemsep=.5ex
\item
\label{CLsCase14a} $b=0$:\quad
${\rm Ch}^0_b=\big\langle\Lambda^4(1),\,\Lambda^1_b(1),\,\Lambda^1_b(t),\,\Lambda^1_b(t^2),\,\Lambda^0,\,\Lambda^2(1),\,\Lambda^2(t),\,\Lambda^3(1),\,\Lambda^3(t)\big\rangle$;
\item
\label{CLsCase14b} $b=x$:\quad
${\rm Ch}^0_b=\big\langle\Lambda^4(1),\,\Lambda^1_b(1),\,\Lambda^1_b(t)+\frac32\Lambda^2(t^2)+\frac12\Lambda^4(t^3),\,\Lambda^1_b(t^2)+\Lambda^2(t^3)+\frac14\Lambda^4(t^4)$,\\[.5ex]
${}$\qquad\qquad$\Lambda^0+\frac12\Lambda^3(t^2),\,\Lambda^2(1)+\Lambda^4(t),\,\Lambda^2(t)+\frac12\Lambda^4(t^2),\,\Lambda^3(1),\,\Lambda^3(t)\big\rangle;$
\item
\label{CLsCase14c} $b=\frac12r^2$:\ \ 
${\rm Ch}^0_b=\big\langle\Lambda^4(1),\Lambda^1_b(1),\Lambda^1_b(e^{2t}),\Lambda^1_b(e^{-2t}),\Lambda^0,\Lambda^2(e^t),\Lambda^2(e^{-t}),\Lambda^3(e^t),\Lambda^3(e^{-t})\big\rangle$;
\item
\label{CLsCase14d} $b=-\frac12r^2$:\quad
${\rm Ch}^0_b=\big\langle\Lambda^4(1),\,\Lambda^1_b(1),\,\Lambda^1_b(\cos 2t),\,\Lambda^1_b(\sin 2t),\,\Lambda^0,\,\Lambda^2(\cos t),\,\Lambda^2(\sin t)$,\\[.5ex]
${}$\qquad\qquad$\Lambda^3(\cos t),\,\Lambda^3(\sin t)\big\rangle$.
\end{enumerate}
\end{enumerate}

\begin{remark}\label{rem:2DSWEsOnSpacesOf0thOrderCCs}
To obtain the associated spaces of zeroth-order conserved currents, 
one needs to replaced each basis conservation-law characteristic 
by the conserved current with the same number and with the same value of the corresponding parameter function 
among those given after the equation~\eqref{eq:RepresentationOfGenCCs}.
\end{remark}

\begin{remark}\label{rem:2DSWEsOnMaxExtensionsOf0thOrderChars}
The spaces of nonpositive-order conservation-law characteristics 
in Cases~\ref{CLsCase2}--\ref{CLsCase6} are indeed maximal extensions 
if the parameter function~$f$ runs only through the values
for which the corresponding values of the arbitrary element~$b$
are not $G^\sim$-equivalent to ones from Cases~\ref{CLsCase7}--\ref{CLsCase14}, 
where the extensions are of greater dimensions.
\end{remark}

\begin{corollary}
The dimension of the space of zeroth-order conservation laws of any system from the class~\eqref{eq:2DSWEs} is not greater than nine.
More specifically, $\dim {\rm Ch}^0_b\in\{2,3,4,5,6,9\}$ for any $b=b(x,y)$.
The span ${\rm Ch}^0_\spanindex$ is wider than the union of all~${\rm Ch}^0_b$, $\bigcup_b{\rm Ch}^0_b\subsetneq{\rm Ch}^0_\spanindex$.
\looseness=-1
\end{corollary}

Similarly to the classification of Lie symmetries, 
the families~$\mathcal T^1$--$\mathcal T^3$ of admissible transformations of the class~\eqref{eq:2DSWEs}
induces additional equivalences among classification cases of Theorem~\ref{thm:ClassificationOf0thOrderCLsOf2DSWEs1},
{
\renewcommand{\labelenumi}{$\mathcal T^\theenumi$:}
\begin{enumerate}\itemsep=0ex\itemindent=2em
\item
\ref{CLsCase3c} $\to$ \ref{CLsCase3a}, \
\ref{CLsCase7c} $\to$ \ref{CLsCase7a}, \
\ref{CLsCase8b} $\to$ \ref{CLsCase8a}$_{\delta=0}$, \
\ref{CLsCase14d}  $\to$ \ref{CLsCase14a};
\item
\ref{CLsCase3b} $\to$ \ref{CLsCase3a}, \
\ref{CLsCase7b} $\to$ \ref{CLsCase7a}, \
\ref{CLsCase8c} $\to$ \ref{CLsCase8a}$_{\delta=0}$, \
\ref{CLsCase14c}  $\to$ \ref{CLsCase14a};
\item
\ref{CLsCase14b}  $\to$ \ref{CLsCase14a}, \quad
\ref{CLsCase4}, \ref{CLsCase8a}, \ref{CLsCase10}, \ref{CLsCase12}:%
\footnote{%
In Cases~\ref{CLsCase10} and~\ref{CLsCase12}, instead of the pure transformational part of $\mathcal T^3$
one should again use its composition with the permutation $(x,u)\leftrightarrow(y,v)$, 
which is associated with a discrete equivalence transformations of the class~\eqref{eq:2DSWEs}, 
cf.\ footnote~\ref{fnt:OnModificationOfAdmTrans}.
}
\ $\delta=1$ $\to$ $\delta=0$.
\end{enumerate}
This is why in Theorem~\ref{thm:ClassificationOf0thOrderCLsOf2DSWEs1} we use, 
analogously to Theorem~\ref{thm:GroupClassificationOf2DSWEs1}, 
the two-level numeration of classification cases, 
cf.\ Remark~\ref{rem:Two-LevelNumerationOfClassificationCases}.
}

In view of the above equivalences, the following assertions are obvious.

\begin{corollary}
A system from the class~\eqref{eq:2DSWEs} admits a nine-dimensional space of zeroth-order conservation laws 
if and only if it is reduced by a point transformation to the system from the same class with~$b=0$.
\end{corollary}

\begin{corollary}
Any system from the class~\eqref{eq:2DSWEs} with five-dimensional space of zeroth-order conservation laws
is similar, with respect to a point transformations, to the system from the same class with~$b=\pm r^{-2}$.
\end{corollary}

\begin{remark}\label{rem:CorrespondenceBetweenLieSymsAnd0thOrderCLsFor2DSWEs}
The comparison of Theorems~\ref{thm:GroupClassificationOf2DSWEs1} and~\ref{thm:ClassificationOf0thOrderCLsOf2DSWEs1}
shows that the extension of the space of zeroth-order conservation laws is possible only for systems 
with Lie-symmetry extension. 
The correspondence of cases of Theorem~\ref{thm:ClassificationOf0thOrderCLsOf2DSWEs1}
to those of Theorem~\ref{thm:GroupClassificationOf2DSWEs1} is obvious, 
and there is a correlation between the dimensions of the corresponding spaces. 
Cases~\ref{LieSymCase1}$_{\nu\ne-2}$, \ref{LieSymCase2}, \ref{LieSymCase4}, \ref{LieSymCase5}, \ref{LieSymCase10}, \ref{LieSymCase11}, 
\ref{LieSymCase12a}--\ref{LieSymCase12c} with $\alpha\ne0$, \ref{LieSymCase13}, \ref{LieSymCase14} and \ref{LieSymCase15} 
of Theorem~\ref{thm:GroupClassificationOf2DSWEs1} have no counterparts among 
cases of Theorem~\ref{thm:GroupClassificationOf2DSWEs1}. 
All these facts can be explained within the framework of the Noether relation 
between Hamiltonian symmetries and conservation laws 
for systems from the class~\eqref{eq:2DSWEs}, see Section~\ref{sec:HamiltoniansOf2DSWEs}.
\end{remark}

Using the above correspondence, we can adopt Conjecture~16 from~\cite{bihl19a}, 
which concerns the group classification of the class~\eqref{eq:2DSWEs} up to the general point equivalence
and was justified by proving the major part of (in)equivalences among classification cases by algebraic tools, 
to zeroth-order conservation laws. 

\begin{conjecture}\label{con:2DSWEsClassificationOfZeroth-orderCLsWrtEquivGroupoid}
A complete list of inequivalent (up to all admissible transformations) cases of extensions 
for the spaces of zeroth-order conservation laws of systems from the class~\eqref{eq:2DSWEs} is exhausted by
Cases~\ref{CLsCase1}, \ref{CLsCase2}, \ref{CLsCase3a}, \ref{CLsCase4}$_{\delta=0}$, \ref{CLsCase5}, \ref{CLsCase6},
\ref{CLsCase7a}, \ref{CLsCase8a}$_{\delta=0}$, \ref{CLsCase9}, \ref{CLsCase10}$_{\delta=0}$, \ref{CLsCase11}, \ref{CLsCase12}$_{\delta=0}$, \ref{CLsCase13}
and~\ref{CLsCase14a} of Theorem~\ref{thm:ClassificationOf0thOrderCLsOf2DSWEs1}.
\end{conjecture}

In view of the reduction of the number of classification cases for zeroth-order conservation laws 
in comparison with that for Lie symmetries, for proving Conjecture~\ref{con:2DSWEsClassificationOfZeroth-orderCLsWrtEquivGroupoid} 
we just need to verify the inequivalence of Cases~\ref{CLsCase1} and~\ref{CLsCase2}
and the impossibility of further gauging of constant parameters remaining in some cases. 
This requires the complete description of the equivalence groupoid of the class~\eqref{eq:2DSWEs}, 
which is still unknown. 

\begin{remark}\label{rem:PhysicalSenseOfCLs}
Among the constructed conservation laws, there are those with obvious physical interpretation. 
Thus, the characteristics 
$\Lambda^0$, $\Lambda^1_b(\frac12)$, $\Lambda^2(1)$, $\Lambda^3(1)$ and $\Lambda^4(1)$ 
corresponds to the conservation of angular momentum, energy, $x$-momentum, $y$-momentum and mass, respectively. 
The pair of the conservation laws with characteristics $\Lambda^2(t)$ and $\Lambda^3(t)$ 
is related to the center-of-mass theorem. 
Note that the zeroth-order conservation laws associated with distinguished (Casimir) functionals 
of the Hamiltonian operator for systems from the class~\eqref{eq:2DSWEs} 
are exhausted by those with characteristics of the form~$\Lambda^4(c)$, where $c$ is an arbitrary constant, 
see Section~\ref{sec:HamiltoniansOf2DSWEs} below. 
\end{remark}

\section{Proof of the classification}\label{sec:ClassificationProof}

For classifying nonpositive-order conservation-law characteristics of systems from the class~\eqref{eq:2DSWEs}, 
the method of furcate splitting works in the same way as for classifying Lie symmetries of such systems,  
i.e., via fixing an arbitrary value for the variable $t$ 
in the classifying equation~\eqref{eq:2DSWEsClassifyingEqFor0thOrderCLs}, 
which results in the so-called \textit{template form} of equations with respect to the arbitrary element~$b$:
\begin{gather}\label{eq:TemplateFormForClassificationOfZeroth-orderCLs}
a_1(xb_x\!+yb_y\!+2b)+a_2(yb_x\!-xb_y)+a_3b_x\!+a_4b_y\!+a_5\frac{x^2\!+y^2}2+a_6x+a_7y+a_8=0.
\end{gather}
Here, $a_1,\dots, a_8$ are constants.

The maximal number $k=k(b)$ of template-form equations
with linearly independent coefficient tuples $\bar a^i=(a^i_1,\dots, a^i_8)$, $i=1,\dots,k$, 
that hold for a fixed value of the arbitrary element~$b$, 
\begin{gather}\label{eq:2DSWEsSysOfTemplateFormEqs}
\begin{split}
&a^i_1(xb_x\!+yb_y\!+2b)+a^i_2(yb_x\!-xb_y)+a^i_3b_x\!+a^i_4b_y\!+a^i_5\frac{x^2\!+y^2}2+a^i_6x+a^i_7y+a^i_8=0,\\
&i=1,\dots,k,
\end{split}
\end{gather}
depends on the fixed value but hereafter we will not indicate this dependence explicitly. 
It is clear that $0\leqslant k\leqslant8$.
If $k>0$, then by 
\[
A:=(a^i_j)_{j=1,\dots,8}^{i=1,\dots,k}, \quad A_l:=(a^i_j)_{j=1,\dots,l}^{i=1,\dots,k},\quad 1\leqslant l\leqslant8
\]
we denote the matrix with coefficients 
of the system of template-form equations~\eqref{eq:2DSWEsSysOfTemplateFormEqs} 
and its submatrix constituted by the first $l$ columns, respectively. 
Thus, $A_8=A$, and $\rank A=k$.
In addition, the consistency of the system~\eqref{eq:2DSWEsSysOfTemplateFormEqs} with respect to~$b$ 
implies $\rank A_4=\rank A=k$, and hence $k\leqslant4$.

To verify this consistency in the case $k>1$,
to the $i$th equation of the system~\eqref{eq:2DSWEsSysOfTemplateFormEqs} for each $i=1,\dots,k$
we relate the vector field
\begin{gather}\label{eq:AssociatedVectorField}
\begin{split}
\mathbf v_i={}&\left(a^i_1x+a^i_2y+a^i_3\right)\p_x+\left(a^i_1y-a^i_2x+a^i_4\right)\p_y\\
&{}-\left(2a^i_1b+\tfrac12a^i_5(x^2+y^2)+a^i_6x+a^i_7y+a^i_8\right)\p_b.
\end{split}
\end{gather}
Then we have $\mathbf v_1,\dots,\mathbf v_k\in\mathfrak a$, where
\[
\mathfrak a:=\big\langle x\p_x+y\p_y+2b\p_b,\,-x\p_y+y\p_x,\,\p_x,\,\p_y,\,(x^2+y^2)\p_b,\,x\p_b,\,y\p_b,\,\p_b\big\rangle.
\]
Direct computation shows that
\begin{gather*}
[\mathbf v_i,\mathbf v_{i'}]\in[\mathfrak a,\mathfrak a]=\big\langle \p_x,\,\p_y,(x^2+y^2)\p_b,\,x\p_b,\,y\p_b,\,\p_b\big\rangle\subset\mathfrak a, \quad i,i'=1,\dots,k,
\end{gather*}
thereby yielding that $[\mathbf v_i,\mathbf v_{i'}]\in\mathfrak a$, $i,i'=1,\dots,k$. 
Therefore, the span $\mathfrak a$ is closed with respect to the Lie bracket of vector fields, thus constituting a Lie algebra.

The equation for $b$ associated with $[\mathbf v_i,\mathbf v_{i'}]$
is a differential consequence of the subsystem of the equations associated with $\mathbf v_i$ and with $\mathbf v_{i'}$,
and it is of the same template form~\eqref{eq:TemplateFormForClassificationOfZeroth-orderCLs}. 
Since the maximal number of linearly independent vector fields
corresponding to template-form equations for the respective value of the arbitrary element~$b$ 
coincides with $k=k(b)$, we obtain
\begin{gather}\label{eq:CompatibilityCondInTermsOfAssociatedVFs}
[\mathbf v_i,\mathbf v_{i'}]\in\big\langle\mathbf v_1,\dots,\mathbf v_k\big\rangle, \quad i,i'=1,\dots,k.
\end{gather}
It is also possible to use the condition~\eqref{eq:CompatibilityCondInTermsOfAssociatedVFs}
formulated for the projections $\hat{\mathbf v}_1$, \dots, $\hat{\mathbf v}_k$ of the vector fields $\mathbf v_1$, \dots, $\mathbf v_k$
onto the space with the coordinates $(x,y)$, which reads
\begin{gather}\label{eq:CompatibilityCondInTermsOfProjectionsOfAssociatedVFs}
[\hat{\mathbf v}_i,\hat{\mathbf v}_{i'}]\in\big\langle\hat{\mathbf v}_1,\dots,\hat{\mathbf v}_k\big\rangle, \quad i,i'=1,\dots,k.
\end{gather}

The system~\eqref{eq:2DSWEsSysOfTemplateFormEqs} can be simplified 
by linearly combining its equations and by acting with transformations from $G^\sim$.
Specifically, in the case $a^i_j\ne0$
one can divide the $i$th equation of~\eqref{eq:2DSWEsSysOfTemplateFormEqs} by~$a^i_j$
to set $a^i_j=1$.
Using simultaneous shifts with respect to~$x$ and~$y$, which belong to~$G^\sim$,
one can set either $a^i_3=a^i_4=0$ if $(a^i_1,a^i_2)\ne(0,0)$ or $a^i_6=a^i_7=0$ if $a^i_5\ne0$.
Similarly, if $a^i_1\ne0$, then shifts of~$b$ allows one to set $a^i_8=0$.

The value $k=0$ corresponds to the case of unrestricted arbitrary element~$b$. 
Here, one can split the equation~\eqref{eq:2DSWEsClassifyingEqFor0thOrderCLs} with respect to $b$ and its derivatives, 
which yields $F^1_t=F^2=F^3=F^4_t=0$ and $c_1=0$. 
Thus, for any value of the arbitrary element~$b$, 
the system~$\mathcal L_b$ admits two linearly independent conservation laws with characteristics~$\Lambda^1_b(1)$ and~$\Lambda^4(1)$, i.e., 
\[
{\rm Ch}^0_b\supseteq{\rm Ch}^{0,\rm unf}_b:=\langle\Lambda^1_b(1),\,\Lambda^4(1)\rangle.
\]
We will call \smash{${\rm Ch}^{0,\rm unf}_b$} 
the space of uniform nonpositive-order conservation-law characteristics of the system~$\mathcal L_b$ 
in the context of the class~\eqref{eq:2DSWEs}. 
This space is two-dimensional. 
The space \smash{${\rm Ch}^0_\cap:=\bigcap_b{\rm Ch}^0_b=\langle\Lambda^4(1)\rangle$} 
of common nonpositive-order conservation-law characteristics of the systems from the class~\eqref{eq:2DSWEs} 
is a proper subspace of \smash{${\rm Ch}^{0,\rm unf}_b$} for each value of the arbitrary element~$b$ 
due to involving~$b$ in $\Lambda^1_b(1)$. 
The space \smash{${\rm Ch}^0_\cap$} has the counterpart, in the sense of commonness, among Lie symmetry structures, 
which is the kernel Lie invariance algebra $\mathfrak g^{\cap}=\langle D(1)\rangle$ of the systems from the class~\eqref{eq:2DSWEs}.
At the same time, there is no counterpart of \smash{${\rm Ch}^{0,\rm unf}_b$}, in the sense of uniformity, 
among Lie symmetry structures related to~$\mathcal L_b$. 
See Section~\ref{sec:HamiltoniansOf2DSWEs} below 
for another relation among the above objects via the Hamiltonian operator~$\mathfrak H$, 
which is common for the systems from the class~\eqref{eq:2DSWEs}. 

The above discussion on \smash{${\rm Ch}^{0,\rm unf}_b$} implies that 
we can reformulate the problem of classification of zeroth-order conservation laws of systems from the class~\eqref{eq:2DSWEs} 
in the following way: 
\emph{find all $G^\sim$-inequivalent values of the arbitrary element~$b$, jointly with~${\rm Ch}^0_b$, 
for which \smash{${\rm Ch}^0_b\supsetneq{\rm Ch}^{0,\rm unf}_b$}}.

For each of the other distinguished cases $k=1$, \dots, $k=4$,
in Section~\ref{sec:DetailsOfProof} below we make the following steps, additionally splitting the consideration into subcases
depending on values of the parameters~$a^i_j$:
\begin{itemize}\itemsep=0ex
\item
for $k>1$, compute the values of the parameters~$a^i_j$
for which the related system~\eqref{eq:2DSWEsSysOfTemplateFormEqs}
is compatible and can be derived by splitting the equation~\eqref{eq:2DSWEsClassifyingEqFor0thOrderCLs},
\item
gauge maximally possible number of the parameters~$a^i_j$
by linearly recombining template-form equations and by acting with transformations from the group~$G^\sim$
and re-denote the remaining parameters~$a^i_j$,
\item
integrate the reduced system~\eqref{eq:2DSWEsSysOfTemplateFormEqs} of template-form equations
with respect to the arbitrary element~$b$,
\item
gauge, if possible, the integration constant by equivalence transformations of the class~\eqref{eq:2DSWEs},
\item
for each of the obtained values of~$b$, 
solve the system of determining equations with respect to the parameters
$c_1$, $F^1$, $F^2$ and $F^3$,
which results in the space~${\rm Ch}^0_b$ of nonpositive-order conservation-law characteristics
of the system~$\mathcal L_b$.
\end{itemize}
The order of steps is not fixed, and they can intertwine. 
In view of~\eqref{eq:RepresentationOfGenCCs}, 
the construction of the corresponding spaces of canonical conserved currents is straightforward.

\section[Hamiltonian structures and generating sets of conservation laws]
{Hamiltonian structures and generating sets\\ of conservation laws}\label{sec:HamiltoniansOf2DSWEs}

Hamiltonian representations for systems from the class~\eqref{eq:2DSWEs} 
can be easily constructed by generalizing the well-known Hamiltonian representation for 
the shallow water equations with flat bottom topography 
\cite[Eqs.~(3.6)--(3.7)]{salm88a}, see also \cite[Section~4.4]{shep90a}. 
For each value of the arbitrary element~$b$, 
the system~$\mathcal L_b$ is Hamiltonian and can be represented in the form 
$\mathbf w_t=\mathfrak H\,\delta\mathcal H_b$, 
where $\delta$~stands for the variational derivative 
with respect to the tuple of dependent variables $\mathbf w=(w^1,w^2,w^3):=(u,v,h)$, 
\[
\mathfrak H:=\begin{pmatrix}
 0&q&-\mathrm D_x\\
-q&0&-\mathrm D_y\\
-\mathrm D_x&-\mathrm D_y&0
\end{pmatrix}, \quad
\mathcal H_b(\mathbf w):=\frac12\iint h\big(u^2+v^2+h-2b\big)\,{\rm d}x\,{\rm d}y 
\] 
are the associated Hamiltonian differential operator and the associated Hamiltonian functional, 
$\mathrm D_x$ and $\mathrm D_y$ denote the total derivative operators with respect to~$x$ and~$y$, 
respectively, and $q:=(v_x-u_y)/h$ is the shallow water potential vorticity. 
Thus, 
\[
\delta\mathcal H_b
=\left(\frac{\delta\mathcal H_b}{\delta u},\frac{\delta\mathcal H_b}{\delta v},\frac{\delta\mathcal H_b}{\delta h}\right)^{\mathsf T}
=\left(hu,hv,\frac12(u^2+v^2)+h-b\right)^{\mathsf T},
\] 
and the Poisson bracket induced by~$\mathfrak H$ is  
$\{\mathcal I,\mathcal J\}=\iint\delta\mathcal I\cdot \mathfrak H\,\delta\mathcal J\,{\rm d}x\,{\rm d}y$
for appropriate functionals~$\mathcal I$ and~$\mathcal J$ of~$\mathbf w$, 
which expands, after an integration by parts, to 
\begin{gather*}
\iint\left(
\left(\mathrm D_x\frac{\delta\mathcal I}{\delta u}+\mathrm D_y\frac{\delta\mathcal I}{\delta v}\right)\frac{\delta\mathcal J}{\delta h}
-\frac{\delta\mathcal I}{\delta h}\left(\mathrm D_x\frac{\delta\mathcal J}{\delta u}+\mathrm D_y\frac{\delta\mathcal J}{\delta v}\right)
+q\left(\frac{\delta\mathcal I}{\delta u}\frac{\delta\mathcal J}{\delta v}-\frac{\delta\mathcal I}{\delta v}\frac{\delta\mathcal J}{\delta u}\right)
\right){\rm d}x\,{\rm d}y.
\end{gather*}
Strengthening results of \cite[Eq.~(4.55)]{shep90a} on distinguished (Casimir) functionals of the Hamiltonian operator~$\mathfrak H$, 
we get the following assertion. 

\begin{proposition}
The space of distinguished (Casimir) functionals of the Hamiltonian operator~$\mathfrak H$ 
consists of the functionals of the form 
$\mathcal C_R:=\iint hR(q)\,{\rm d}x\,{\rm d}y$, where $R$ is an arbitrary smooth function of~$q$. 
\end{proposition}

\begin{proof}
Denoting $R':=\p R/\p q$, for any~$R$ one has 
\[\mathfrak H\delta\mathcal C_R=\mathfrak H\big(\mathrm D_yR'(q),-\mathrm D_xR'(q),R(q)-qR'(q)\big)^{\mathsf T}=0,\]
i.e., $\mathcal C_R$ is a distinguished (Casimir) functional for~$\mathfrak H$.

Conversely, let a functional $\mathcal K=\iint K{\rm d}x\,{\rm d}y$ be distinguished for~$\mathfrak H$. 
This means by definition that $K=K[u,v,h]$ is a differential function that depends at most on 
$(x,y)$, the dependent variables $(u,v,h)$ and their derivatives with respect to $(x,y)$, 
and satisfies the equation $\mathfrak H\mathsf EK=0$. 
Here $\mathsf E$ denotes the Euler operator with respect to $(u,v,h)$. 
Then the tuple of differential functions $\mathsf EK=:(L^1,L^2,L^3)$ belongs to the kernel of~$\mathfrak H$, 
which implies $L^1=-q^{-1}\mathrm D_yL^3$, $L^2=q^{-1}\mathrm D_xL^3$ and $(\mathrm D_xq)\mathrm D_yL^3-(\mathrm D_yq)\mathrm D_xL^3=0$.
If we suppose that $\ord L^3>1$, then the last equation gives a contradiction. 
Therefore, $\ord L^3\leqslant1$. 
Similarly, we prove that $L^3_{h_x}=0$, i.e., the function~$L^3$ depends at most on $(t,x,y,u,v,q,u,v,u_x,u_y,v_x,v_y)$,
where we change the jet coordinates, replacing~$h$ by~$q$. 
Representing the same equation in the new coordinates, we derive that $L^3=L^3(q)$. 
Modulo adding total divergences with respect to $(x,y)$, 
the solution of the equation $\mathsf EK=(L^1,L^2,L^3)$ for~$K$ leads to the representation 
$K=hR(q)$, where $R$ is an arbitrary solution of the equation $R-qR'=L^3$.
\end{proof}

The local conservation law associated with the functional~$\mathcal C_R$ has the characteristic 
\[\big(\mathrm D_yR'(q),-\mathrm D_xR'(q),R(q)-qR'(q)\big)^{\mathsf T}\] 
(cf. the beginning of the above proof) and contains the conserved current $hR(q)(1,u,v)$. 
Therefore, this conservation law is zero if and only if $R=cq$ for some constant~$c$. 
Its order is equal to zero or one if $R$ is a nonzero constant or $R''\ne0$, respectively. 
Among elements of the family $\{\mathcal C_R\}$, there are functionals associated 
with the conservation of mass ($R=1$), 
the trivial conservation of circulation ($R=q$) 
and the conservation of potential enstrophy ($R=\frac12q^2$).

The operator~$\mathfrak H$ maps cosymmetries and conservation-law characteristics of~$\mathcal L_b$ 
to characteristics of symmetries and of Hamiltonian symmetries, respectively. 
Acting by~$\mathfrak H$ on the elements~\eqref{eq:SpanningCanditatesForChars} of~${\rm Ch}^0_\spanindex$,
we obtain 
\begin{gather*}
\mathfrak H\Lambda^0=J[\mathbf w],\quad
\mathfrak H\Lambda^1_b(F^1)=2D(F^1)[\mathbf w]\quad\mbox{on}\quad \mathcal L_b,\\
\mathfrak H\Lambda^2(F^2)=P(F^2,0)[\mathbf w],\quad
\mathfrak H\Lambda^3(F^3)=P(0,F^3)[\mathbf w],\quad
\mathfrak H\Lambda^4(F^4)=(0,0,0)^{\mathsf T}.
\end{gather*}
Here the characteristic~$\mathbf v[\mathbf w]$ of a vector field~$\mathbf v=\tau\p_t+\xi^1\p_x+\xi^2\p_y+\eta^1\p_{w^1}+\eta^2\p_{w^2}+\eta^3\p_{w^3}$ 
with components depending on $(t,x,y,u,v,h)$ is given by \[\mathbf v[\mathbf w]=(\eta^i-\tau w^i_t-\xi^1 w^i_x-\xi^2 w^i_y,i=1,2,3).\]
The tuple $\Lambda^4(F^4)$ belongs to the kernel of~$\mathfrak H$ since 
it is the variational derivative of the time-dependent distinguished functional $F^4\mathcal C_1$.
At the same time, the tuple 
\[\big(t(uu_x+vu_y+h_x-b_x)-u,\,t(uv_x+vv_y+h_y-b_y)-v,\,t(uh)_x+t(vh)_y-2h\big)^{\mathsf T}\]
coinciding with $D^{\rm t}[\mathbf w]$ on solutions of $\mathcal L_b$
does not belong to the image of~$\mathfrak H$ under acting on triples of differential functions
since, in particular, $h$ cannot be equal to the total divergence of a pair of differential functions of~$\mathbf w$ 
with respect to the space variables~$(x,y)$.

This means that the algebra~$\mathfrak h_b=\mathfrak H{\rm Ch}^0_b$ of Hamiltonian Lie symmetries 
of the system~$\mathcal L_b$ coincides 
with the intersection $\mathfrak g_b\cap\mathfrak h_\spanindex$, 
where \[\mathfrak h_\spanindex:=\langle D(F^1),\,J,\,P(F^2,F^3)\rangle\] 
with the parameter functions~$F^1$, $F^2$ and $F^3$ running through the set of smooth functions of~$t$. 
Note that $\mathfrak h_b$ is an ideal of $\mathfrak g_b$ with $\dim\mathfrak h_b=\dim{\rm Ch}^0_b-1$, and 
$\mathfrak H{\rm Ch}^{0,\rm unf}_b=\mathfrak g^\cap$, 
$\mathfrak H{\rm Ch}^0_\cap=\{0\}$.

Using $\mathfrak H$ as a Noether operator, we can endow the space~${\rm Ch}_b$ of cosymmetries of the system~$\mathcal L_b$ 
with the structure of Lie algebra, where the Lie bracket $[\cdot,\cdot]_{\mathfrak H}$ is defined by 
\[
[\gamma^1,\gamma^2]_{\mathfrak H}
=\mathsf D_{\gamma^2}\mathfrak H\gamma^1
+\mathsf D_{\mathfrak H\gamma^1}^\dag\gamma^2
+(\mathsf D_{\gamma^1}-\mathsf D_{\gamma^1}^\dag)\mathfrak H\gamma^2
\]
for the canonical representatives~$\gamma^1$ and~$\gamma^2$ of arbitrary cosymmetry cosets of~$\mathcal L_b$, 
cf., e.g., \cite{fuch82b} and \cite[Section~3.1]{blas98a}. 
(The canonical representative of a cosymmetry coset of a system of evolution equations 
is the element of the coset that depends on no derivatives involving differentiation with respect to~$t$.) 
Recall that $\mathsf D_\gamma$ and \smash{$\mathsf D_\gamma^\dag$} denote the Fr\'echet derivative of~$\gamma$ 
and its formal adjoint, respectively. 
Formally commuting the elements~\eqref{eq:SpanningCanditatesForChars} of~${\rm Ch}^0_\spanindex$, 
we get 
\begin{gather*}
[\Lambda^0,\Lambda^1_b(F^1)]_\mathfrak H=2F^1(0,0,xb_y-yb_x),
\\[.5ex]
[\Lambda^0,\Lambda^2(F^2)]_\mathfrak H=-\Lambda^3(F^2),\quad
[\Lambda^0,\Lambda^3(F^3)]_\mathfrak H= \Lambda^2(F^3),\quad
\\[.5ex]
[\Lambda^1_b(F^1),\Lambda^1_b(\tilde F^1)]_\mathfrak H=2\Lambda^1_b(F^1\tilde F^1_t-\tilde F^1F^1_t)\\
\qquad{}+\big(0,0,-2(F^1\tilde F^1_t-\tilde F^1F^1_t)(xb_x+yb_y+2b)+(F^1\tilde F^1_{ttt}-\tilde F^1F^1_{ttt})(x^2+y^2)\big),
\\[.5ex]
[\Lambda^1_b(F^1),\Lambda^2(F^2)]_\mathfrak H=-2F^1(0,0,F^2b_x-F^2_{tt}x),
\\[.5ex]
[\Lambda^1_b(F^1),\Lambda^3(F^3)]_\mathfrak H=-2F^1(0,0,F^3b_y-F^3_{tt}y),
\\[.5ex]
[\Lambda^0,\Lambda^4(F^4)]_\mathfrak H=[\Lambda^2(F^2),\Lambda^3(F^3)]_\mathfrak H
=[\Lambda^2(F^2),\Lambda^4(F^4)]_\mathfrak H=[\Lambda^3(F^3),\Lambda^4(F^4)]_\mathfrak H=(0,0,0).
\end{gather*}
Therefore, the space~${\rm Ch}^0_b$ of characteristics of zeroth-order conservation laws of~$\mathcal L_b$
is closed with respect to the above bracket 
in view of the classifying equation~\eqref{eq:2DSWEsClassifyingEqFor0thOrderCLs}. 
Moreover, for any \mbox{$\gamma^1,\gamma^2\in{\rm Ch}^0_b$} we have
\smash{$[\gamma^1,\gamma^2]_{\mathfrak H}=\mathsf D_{\gamma^2}\mathfrak H\gamma^1+\mathsf D_{\mathfrak H\gamma^1}^\dag\gamma^2$} 
since \smash{$\mathsf D_\gamma=\mathsf D_\gamma^\dag$} for the canonical representative~$\gamma$ of any coset 
of conservation-law characteristics of~$\mathcal L_b$. 
In other words, the operator~$\mathfrak H$ gives rise to a homomorphism from the Lie algebra~${\rm Ch}^0_b$ 
to the Lie algebra~$\mathfrak g_b$ with the kernel $\langle\Lambda^4(1)\rangle$ 
and the image $\mathfrak h_b=\mathfrak g_b\cap\mathfrak h_\spanindex$.

The above consideration establishes the Noether relation between 
Hamiltonian Lie symmetries and zeroth-order conservation laws
for each system from the class~\eqref{eq:2DSWEs}. 
Moreover, this consideration explains the facts given 
in Remark~\ref{rem:CorrespondenceBetweenLieSymsAnd0thOrderCLsFor2DSWEs}.
An extension of the space~${\rm Ch}^0_b$  
corresponds to an extension of the algebra~$\mathfrak h_b$ of Hamiltonian Lie symmetries 
rather than an extension of the entire algebra~$\mathfrak g_b$ of Lie symmetries. 
For values of the arbitrary element~$b$, where $\mathfrak h_b$ is a proper subalgebra of~$\mathfrak g_b$, 
the correspondence between~$\mathfrak g_b$ and~${\rm Ch}^0_b$ is broken. 
Since in Cases~\ref{LieSymCase1}$_{\nu\ne-2}$, \ref{LieSymCase2}, \ref{LieSymCase4}, \ref{LieSymCase5}, \ref{LieSymCase10}, \ref{LieSymCase11}, 
\ref{LieSymCase12a}--\ref{LieSymCase12c} with $\alpha\ne0$, \ref{LieSymCase13}, \ref{LieSymCase14} and \ref{LieSymCase15} 
of Theorem~\ref{thm:GroupClassificationOf2DSWEs1}, 
vector fields extending the corresponding algebra~$\mathfrak g_b$ involve~$D^{\rm t}$ and are thus not Hamiltonian, 
these cases have no counterparts among cases of Theorem~\ref{thm:ClassificationOf0thOrderCLsOf2DSWEs1}. 
In Cases~\ref{LieSymCase17}--\ref{LieSymCase22} of Theorem~\ref{thm:GroupClassificationOf2DSWEs1}, 
the algebra $\mathfrak h_b$ extends although it is a proper subalgebra of~$\mathfrak g_b$. 
Hence these cases have counterparts among cases of Theorem~\ref{thm:ClassificationOf0thOrderCLsOf2DSWEs1} 
but $\dim\mathfrak g_b>\dim{\rm Ch}^0_b-1$.
 
One could single out Hamiltonian symmetries among the Lie symmetries listed 
in Theorem~\ref{thm:GroupClassificationOf2DSWEs1} 
and then use them for constructing zeroth-order conservation laws of systems from the class~\eqref{eq:2DSWEs} 
within the framework of the Noether relation between Hamiltonian symmetries and conservation laws 
\cite[Theorem~7.15]{olve93a}. 
This would allow to classify zeroth-order conservation laws
for systems from the class~\eqref{eq:2DSWEs} 
using the group classification of this class in Theorem~\ref{thm:GroupClassificationOf2DSWEs1}.
At the same time, the direct classification of such conservation laws 
in Section~\ref{sec:ClassificationProof} seems to be a more efficient way. 

Now for each value of the arbitrary element~$b$, 
we study the action of Lie symmetries of the system~$\mathcal L_b$ on its zeroth-order conservation laws, 
more precisely, on characteristics of such conservation laws.%
\footnote{\label{fnt:ActioOfGenSymsOnCosyms}%
\newcommand{\EqOrd}{r}
Let~$\mathcal L$ be a system of $p$ differential equations $L^s(x,u_{(\EqOrd)})=0$, $s=1,\dots,p$,
for $m$ unknown functions $u=(u^1,\ldots,u^m)$
of $n$ independent variables $x=(x_1,\ldots,x_n)$.
Here $u_{(\EqOrd)}$ denotes the set of all the derivatives of the functions $u$ with respect to $x$
of order not greater than~$\EqOrd$, including $u$ as the derivative of order zero.
Suppose that a generalized vector field~$Q=\sum_{i=1}^n\xi^i[u]\p_{x_i}+\sum_{a=1}^m\eta^a[u]\p_{u^a}$ is a generalized symmetry of~$\mathcal L$, 
and a tuple $\gamma=(\gamma^1[u],\dots,\gamma^p[u])^{\mathsf T}$ is a cosymmetry 
(or, more specifically, a conservation-law characteristic) of~$\mathcal L$. 
Here $\xi^i[u]$, $\eta^a[u]$ and $\gamma^s[u]$ are differential functions of~$u$.
Therefore, $Q_{(\EqOrd)}L^s=\sum_{s'=1}^p\Upsilon^{ss'}L^{s'}$
for some matrix operator $\Upsilon=(\Upsilon^{ss'})_{s,s'=1,\dots,p}$ 
whose entries are differential operators in total derivatives with respect to~$x$
with coefficients being differential functions of~$u$. 
Here $Q_{(\EqOrd)}$ is the $r$th prolongation of~$Q$. 
See \cite{olve93a} for definitions. 
The action of~$Q$ on~$\gamma$ is given by 
\[\mathfrak L_Q\gamma=(Q_{(\infty)}+\mathop{\rm Div}\xi)\gamma+\Upsilon^\dag\gamma,\]
where $\mathop{\rm Div}\xi:=\sum_{i=1}^n\mathrm D_i\xi^i$ is the total divergence of $\xi:=(\xi^1[u],\dots,\xi^n[u])$,
$\mathrm D_i$ is the total derivative operator with respect to~$x_i$, $i=1,\dots,n$, 
and the matrix operator~$\Upsilon^\dag$ is the formally adjoint of~$\Upsilon$.
}
Thus, we can construct a generating set for the corresponding space, 
which leads to a generating set for zeroth-order conservation laws.

Let $Q=2D(F^1)-c_1D^{\rm t}-c_2J+P(F^2,F^3)$ with functions~$F^1$, $F^2$ and $F^3$ of~$t$ and constants~$c_1$ and~$c_2$
be a Lie symmetry vector field of the system~$\mathcal L_b$ 
from the class~\eqref{eq:2DSWEs} for a fixed value of the arbitrary element~$b$, $Q\in\mathfrak g_b$.
Then, in the notation of footnote~\ref{fnt:ActioOfGenSymsOnCosyms}, 
$\mathop{\rm Div}\xi=4F^1_t-c_1$, $Q_{(1)}L^s=\sum_{s'=1}^3\Upsilon^{ss'}L^{s'}$, $s=1,2,3$, 
where $(L^1,L^2,L^3)$ is the left-hand side of~$\mathcal L_b$ and 
\[
\Upsilon=-F^1_t\mathop{\rm diag}(3,3,4)+c_1\mathop{\rm diag}(2,2,3)-c_2\begin{pmatrix}0&-1\\1&0\end{pmatrix}\oplus(0).
\]  
Hence the action of~$Q$ on a nonpositive-order cosymmetry~$\gamma$ of~$\mathcal L_b$ is given by 
\[
\mathfrak L_Q\gamma=Q\gamma+
\begin{pmatrix}F^1_t+c_1&-c_2&0\\c_2&F^1_t+c_1&0\\0&0&2c_1\end{pmatrix}\gamma.
\]
Acting by $\mathfrak L_Q$ on a tuple 
$\gamma=\tilde c_1\Lambda^0+\Lambda^1_b(\tilde F^1)+\Lambda^2(\tilde F^2)+\Lambda^3(\tilde F^3)+\Lambda^4(\tilde F^4)\in{\rm Ch}^0_b$ 
with functions~$\tilde F^1$, $\tilde F^2$, $\tilde F^3$ and $\tilde F^4$ of~$t$ and a constant~$\tilde c_1$, 
we get
\[
\begin{split}
\mathfrak L_Q\gamma={}&3c_1\gamma+
  \Lambda^1_b\big(2(F^1\tilde F^1_t-F^1_t\tilde F^1)-c_1(t\tilde F^1_t-\tilde F^1)\big)\\
&+\Lambda^2\big(2F^1\tilde F^2_t-F^1_t\tilde F^2-2\tilde F^1F^2_t+\tilde F^1_tF^2-c_1t\tilde F^2_t-c_2\tilde F^3+\tilde c_1F^3\big)\\
&+\Lambda^3\big(2F^1\tilde F^3_t-F^1_t\tilde F^3-2\tilde F^1F^3_t+\tilde F^1_tF^3-c_1t\tilde F^3_t+c_2\tilde F^2-\tilde c_1F^2\big)\\
&+\Lambda^4\big(2F^1\tilde F^4_t-c_1t\tilde F^4_t+F^2\tilde F^2_t-\tilde F^2F^2_t+F^3\tilde F^3_t-\tilde F^3F^3_t-c_1\tilde F^4-2F^4\tilde F^1\big),
\end{split}
\]
where the function $F^4=F^4(t)$ is defined by the classifying equation~\eqref{eq:2DSWEsClassifyingEq} 
for given $F^1$, $F^2$, $F^3$, $c_1$, $c_2$ and~$b$. 
Of course, the action of~$Q$ on~$\gamma$ agrees with 
the Noether relation between Hamiltonian symmetries and conservation laws, 
cf.\ the commutation relations~\eqref{eq:2DSWEsCommutationRelations}. 
More specifically, 
$\mathfrak H\mathfrak L_Q\gamma=[Q,\mathfrak H\gamma]+3c_1\mathfrak H\gamma$, 
recalling that a Lie symmetry of~$\mathcal L_b$ is not Hamiltonian if and only if it involves~$D^{\rm t}$.

In view of Theorems~\ref{thm:GroupClassificationOf2DSWEs1} and~\ref{thm:ClassificationOf0thOrderCLsOf2DSWEs1} 
and the found additional equivalences among classification cases, 
it suffices to construct (minimal) generating sets of elements of~${\rm Ch}^0_b$ under action of~$\mathfrak g_b$
for the general value of~$b$ 
and for cases listed in Conjecture~\ref{con:2DSWEsClassificationOfZeroth-orderCLsWrtEquivGroupoid}. 
Thus, generating sets are the following:

\par\medskip\par
General case: $\{\Lambda^4(1),\Lambda^1_b(1)\}$;
\par\medskip\par
\ref{CLsCase1}: $\{\Lambda^4(1),\Lambda^1_b(t)+\beta\Lambda^0\}$;\quad 
\ref{CLsCase2}$_{\delta=0}$: $\{\Lambda^4(1),\Lambda^1_b(1),\Lambda^0\}$;\quad 
\ref{CLsCase2}$_{\delta\ne0}$: $\{\Lambda^1_b(1),\Lambda^0+\delta\Lambda^4(t)\}$;
\par\medskip\par
\ref{CLsCase3a}: $\{\Lambda^4(1),\Lambda^1_b(1)\}$;\quad
\ref{CLsCase4}$_{\delta=0}$: $\{\Lambda^1_b(1),\Lambda^2(t)\}$;\quad 
\ref{CLsCase7a}: $\{\Lambda^4(1),\Lambda^1_b(1),\Lambda^0\}$;\quad 
\ref{CLsCase8a}$_{\delta=0}$: $\{\Lambda^1_b(1)\}$;
\par\medskip\par
\ref{CLsCase10}$_{\delta=0}$, \ref{CLsCase12}$_{\delta=0}$: $\{\Lambda^1_b(1),\Lambda^3(t)\}$;\quad
\ref{CLsCase5}, \ref{CLsCase6}, \ref{CLsCase9}, \ref{CLsCase11}, \ref{CLsCase13}: $\{\Lambda^1_b(1)\}$;\quad
\ref{CLsCase14a}: $\{\Lambda^1_b(1),\Lambda^0\}$. 

\begin{remark}\label{rem:IbragimovResults}
In~\cite[Section~25.4]{ibra85A}, the basis zeroth-order conservation laws 
of the system~$\mathcal L_0$ of shallow water equations with flat bottom topography 
(Case~\ref{CLsCase14a} of Theorem~\ref{thm:ClassificationOf0thOrderCLsOf2DSWEs1})
were generated via acting by Lie symmetries on the conservation laws of energy and angular momentum, 
which have the characteristics~$\Lambda^1_0(1)$ and~$\Lambda^0$, respectively.%
\footnote{%
More precisely, this procedure was realized for the equations of motion of an ideal polytropic gas 
and then supplemented by the observation that the reduction of these equations 
to the case of two-dimensional isotropic gas flows with adiabatic exponent two coincides, 
after re-interpreting the meaning of the dependent variables, with the system~$\mathcal L_0$. 
}
Nevertheless, no argument was presented that the constructed conservation laws span 
the entire space of zeroth-order conservation laws of~$\mathcal L_0$. 
As was noted in~\cite[Section~25.4]{ibra85A}, 
a common property of mechanical models that are invariant with respect to the Galilean group 
is that two conservation laws are basic for them, the conservation of energy and angular momentum. 
The system~$\mathcal L_0$ shares this property in the above sense.
\end{remark}

\begin{remark}\label{rem:CLsComparisonWithEulerEqs}
Both the two- and three-dimensional (incompressible) Euler equations also admit Hamiltonian structures~\cite{olve82a} 
although these structures are not standard since they are formulated in terms of the vorticity 
but not the original dependent variables, which are the velocity and the pressure; 
see additionally \cite[Example~7.10]{olve93a}. 
Up to linearly combining, all the known conservation laws for these equations are exhausted 
by zeroth-order conservation laws associated with their Hamiltonian Lie symmetries via the Noether relation, 
first-order conservation laws associated with distinguished (Casimir) functionals of the corresponding Hamiltonian operators 
and zeroth-order conservation laws with zero densities that are related to the incompressibility condition; 
cf.~\cite{olve82a} and \cite[Example~7.17]{olve93a}. 
It was conjectured in~\cite{olve82a} and is still not proved 
that the Euler equations possess no proper (i.e., positive-order) generalized symmetries 
and thus, the entire space of their local conservation laws is spanned by the known ones. 
We conjecture that the similar assertion holds for the shallow water equations with an arbitrary bottom topography as well.  
These conjectures are in particular supported by the description of generalized symmetries and local conservation laws 
of the three-dimensional (incompressible) Navier--Stokes equations, which was obtained in~\cite{gusy89b}. 
More specifically, it was proved therein that any generalized symmetry and any conservation-law characteristic of these equations 
are equivalent to a Lie symmetry and to a conservation-law characteristic of order~$-\infty$, respectively.
Comprehensive lists of independent local conservation laws were presented 
in~\cite{chev14a} for the three-dimensional \mbox{Euler} and Navier--Stokes equations
and in~\cite{kelb13a} for their reductions, including the two--dimensional Euler equations. 
The completeness of these lists for conservation laws with characteristics up to order two and zero, respectively, 
was shown therein using the package {\sf GeM}~\cite{chev07a}. 
(We had checked such completeness for the (two-dimensional) vorticity equation up to order four 
using some tricks and the package {\sf Jets} \cite{BaranMarvan,marv09a}.)
The space of distinguished (Casimir) functionals for the two--dimensional Euler equations 
is parameterized by an arbitrary function of the vorticity, similarly to  the shallow water equations. 
(In space dimension three, the corresponding space is spanned by the single functional associated with the conservation of helicity.) 
Modulo the conservation laws originated by the incompressibility condition, 
the comparison of zeroth-order conservation laws of the two--dimensional Euler equations 
and of the systems from the class~\eqref{eq:2DSWEs} is similar, due to the presence of Hamiltonian structures, 
to the discussion of Lie symmetries in Remark~\ref{rem:LieSymsComparisonWithEulerEqs}. 
\end{remark}

\section{Conclusion}\label{sec:ConclusionsSWE}

We classified zeroth-order conservation laws
of two-dimensional shallow water equations with variable bottom topography, 
which are of the form~\eqref{eq:2DSWEs}, 
up to the equivalence generated by the action of the associated equivalence group $G^\sim$.
The classification is summarized in Theorem~\ref{thm:ClassificationOf0thOrderCLsOf2DSWEs1}. 
To prove it, we used the version of the method of furcate splitting suggested in~\cite{bihl19a}. 
In fact, this is the first application of the method of furcate splitting
to classifying conservation laws in the literature. 
The template form~\eqref{eq:TemplateFormForClassificationOfZeroth-orderCLs} 
arising in the course of the classification 
is a specification, by the margin of one free coefficient, of the template form for furcate splitting 
in the course of the group classification of the class~\eqref{eq:2DSWEs}. 
This is why we just modified the proof of Theorem~\ref{thm:GroupClassificationOf2DSWEs1}
presented in~\cite{bihl19a}.

\looseness=1
The relation between the template forms is a display 
of the Noether relation between Hamiltonian symmetries and conservation laws 
for systems from the class~\eqref{eq:2DSWEs}.
The Hamiltonian nature of these systems
also allowed us to relate the classification lists for their Lie symmetries and for their zeroth-order conservation laws, 
which gave one more check of both the classifications and contributed to a deeper understanding of them.
Due to the possibility of following the proof of group classification, 
the direct construction of zeroth-order conservation laws of systems from the class~\eqref{eq:2DSWEs}
is of the same order of computational complexity as 
the construction based on the Noether relation between Hamiltonian symmetries and conservation laws. 
An incidental benefit of the direct construction is a test of furcate splitting
as a method for classifying low-order conservation laws 
that can be applied to classes of non-Hamiltonian systems as~well. 

A Lie structure underlying the template-form equations for conservation laws of each system from the class~\eqref{eq:2DSWEs}
provided us with algebraic tools for optimizing the procedure of furcate splitting.  
It is yet not clear whether this Lie structure is specific for systems from the class~\eqref{eq:2DSWEs} 
and induced by their Hamiltonian structure  
or whether similar Lie structures may arise in the course of classification of conservation laws 
for non-Hamiltonian systems.

Along with zeroth-order conservation laws, 
each system from the class~\eqref{eq:2DSWEs} admits 
an infinite number of linearly independent first-order conservation laws 
associated with distinguished (Casimir) functionals of the Hamiltonian operator~$\mathfrak H$, 
which is common for all such systems. 
An interesting question is whether the above conservation laws span 
the entire space of local conservation laws for any system from the class~\eqref{eq:2DSWEs}.

The present paper has focused exclusively on deriving the zeroth-order conservation law classification. 
We hope that these results prove relevant for practitioners developing numerical models for tsunami propagation 
and other applications in atmosphere--ocean science. 
Conservation laws provide constraints that have to be satisfied by any solution 
of the corresponding system of differential equations. 
Therefore, monitoring the numerical conservation over time can provide important feedback 
for assessing the quality of newly developed numerical models to these equations. 

We believe that the present paper jointly with~\cite{bihl19a} can also be regarded as a stepping stone 
towards the realization of a more comprehensive research program 
devoted to studying the geometric properties of the shallow water equations. 
While in~\cite{bihl19a} we have solved the complete group classification problem for the class~\eqref{eq:2DSWEs}, 
it still remains to compute the equivalence groupoid of this class 
and to carry out inequivalent Lie reductions for constructing the associated exact group-invariant solutions. 
The simultaneous reduction of conservation laws and Hamiltonian structures may be relevant for the latter problem. 
A related task that should be completed in the course of consideration of systems from the class~\eqref{eq:2DSWEs} 
within the framework of group analysis of differential equations is the study of hidden symmetries of these systems. 
Here one investigates whether the reduced systems admit additional symmetries not inherited from the symmetries of the original systems. 
These hidden symmetries could then be used for finding additional group-invariant solutions of~\eqref{eq:2DSWEs}. 
The motivation behind this investigation would again be to construct new exact solutions to~\eqref{eq:2DSWEs}, 
which could then be used as reference solutions for testing numerical models of the shallow water equations. 
The study of hidden symmetries can be extended to other hidden objects, 
including hidden conservation laws and hidden Hamiltonian structures. 
We have started to implement the above research program for the shallow water equations with flat bottom topography.

\appendix

\section{Details of the proof using the method of furcate splitting}
\label{sec:DetailsOfProof}

In this section, we present details of applying the method of furcate splitting 
to classifying zeroth-order conservation laws of systems from the class~\eqref{eq:2DSWEs}. 
See Section~\ref{sec:ClassificationProof} for notations.

\subsection{One independent template-form equation}

If $k=1$, then the right-hand side of the classifying equation~\eqref{eq:2DSWEsClassifyingEqFor0thOrderCLs}
is proportional to the right-hand side of the single template-form equation~\eqref{eq:2DSWEsSysOfTemplateFormEqs}, 
where the proportionality coefficient~$\lambda$ is a sufficiently smooth function of~$t$, i.e.
\begin{gather*}
F^1_t(xb_x+yb_y+2b)+c_1(yb_x-xb_y)+F^2b_x+F^3b_y-F^1_{ttt}\frac{x^2+y^2}{2}-F^2_{tt}x-F^3_{tt}y-F^4_t\\
=\lambda\bigg(a^1_1(xb_x+yb_y+2b)+a^1_2(yb_x-xb_y)+a^1_3b_x+a^1_4b_y+a^1_5\frac{x^2+y^2}2+a^1_6x+a^1_7y+a^1_8\bigg).
\end{gather*}
The function~$\lambda$ is nonvanishing for any tuple 
from the complement of~\smash{${\rm Ch}^{0,\rm unf}_b$} in~${\rm Ch}^0_b$.
Splitting the last equation with respect to the involved derivatives of $b$ (including~$b$ itself) 
and the independent variables~$x$ and~$y$ yields the system
\begin{gather}\label{k1}
\begin{split}
&F^1_t=a^1_1\lambda,\quad c_1=a^1_2\lambda,\quad F^2=a^1_3\lambda,\quad F^3=a^1_4\lambda,\\
&F^1_{ttt}=-a^1_5\lambda,\quad F^2_{tt}=-a^1_6\lambda,\quad F^3_{tt}=-a^1_7\lambda,\quad F^4_t=-a^1_8\lambda.
\end{split}
\end{gather}
Requiring that $\rank A_4=k=1$ in the present case implies that $(a^1_1,a^1_2,a^1_3,a^1_4)\ne(0,0,0,0)$.
The consideration therefore splits into the following three cases, which we consider subsequently:
\[
a^1_1\ne0;\quad a^1_1=0,\ a^1_2\ne0;\quad a^1_1=a^1_2=0,\ (a^1_3,a^1_4)\ne(0,0).
\]

\noindent$\boldsymbol{a^1_1\ne0.}$
First we rescale the equation~\eqref{eq:2DSWEsSysOfTemplateFormEqs} to set $a^1_1=1$. 
Then we gauge other coefficients of~\eqref{eq:2DSWEsSysOfTemplateFormEqs} 
by transformations from $G^\sim$ for further simplifying the computations. 
Thus, we set $a^1_3=a^1_4=0$ and $a^1_8=0$ with point equivalence transformations of simultaneous shifts with respect to $(x,y)$ 
and of shifts with respect to $b$, respectively. 
Under the imposed constraints, we get $F^2=0$, $F^3=0$ and $F^4_t=0$,
and the system~\eqref{k1} yields $a^1_6=a^1_7=0$.
In the polar coordinates~$(r,\varphi)$, the equation~\eqref{eq:2DSWEsSysOfTemplateFormEqs} reduces to 
\[
rb_r-a^1_2b_\varphi+2b+\frac12a^1_5r^2=0,
\]
and its solution depends on values of the parameters $a^1_2$ and~$a^1_5$.

For $a^1_2=0$, we can make $a^1_5\in\{0,-4,4\}$ by scaling equivalence transformations,
which produces Cases \ref{CLsCase3a}, \ref{CLsCase3b} and \ref{CLsCase3c}
of Theorem~\ref{thm:ClassificationOf0thOrderCLsOf2DSWEs1}, respectively.

For $a^1_2\ne0$, due to the system~\eqref{k1} we find $c_1=a^1_2\lambda$.
Therefore, $\lambda$ is a constant, and hence $F^1_t$ is also a constant, which implies $a^1_5=0$,
leading to Cases~\ref{CLsCase1}. 

\medskip\par\noindent$\boldsymbol{a^1_1=0,\, a^1_2\ne0.}$
By rescaling the equation~\eqref{eq:2DSWEsSysOfTemplateFormEqs} as well as  
shifting $x$ and $y$ and scaling variables,
we can assume $a^1_2=1$, $a^1_3=a^1_4=0$ and $a^1_8\in\{0,1\}$.
These conditions for~$a$'s jointly with the system~\eqref{k1} imply that
$F^1_t=0$, $F^2=F^3=0$, $\lambda=c_1$, and thus $a^1_5=a^1_6=a^1_7=0$.
After integrating the equation~\eqref{eq:2DSWEsSysOfTemplateFormEqs} 
using its representation $b_\varphi=a^1_8$ in the polar coordinates~$(r,\varphi)$, 
we get Case~\ref{CLsCase2}.

\par\medskip\par\noindent $\boldsymbol{a^1_1=a^1_2=0,\, (a^1_3,a^1_4)\ne(0,0).}$
Rotation equivalence transformations
and rescaling the equation~\eqref{eq:2DSWEsSysOfTemplateFormEqs} 
allow rotating and scaling the tuple $(a^1_3,a^1_4)$ to set $a^1_3=1$ and $a^1_4=0$.
From the system~\eqref{k1}, we obtain $F^1_t=F^3=0$, $c_1=0$, 
$F^2=\lambda$ and therefore $a^1_5=a^1_7=0$ and $F^2_{tt}=-a^1_6F^2$.
Thus, the template-form equation~\eqref{eq:2DSWEsSysOfTemplateFormEqs} takes the form
$b_x+a^1_6x+a^1_8=0$.
We can set the constraint $a^1_6\in\{0,-1,1\}$ by using scaling equivalence transformations,
which gives Cases~\ref{CLsCase4}, \ref{CLsCase5} and~\ref{CLsCase6}.
Note that $a^1_8=0\bmod G^\sim$ if $a^1_6\ne0$
and $a^1_8\in\{0,-1\}\bmod G^\sim$ if $a^1_6=0$.

\subsection{Two independent template-form equations}

If $k=2$, then the right-hand side of the classifying equation~\eqref{eq:2DSWEsClassifyingEqFor0thOrderCLs} is represented as 
a linear combination of right-hand sides of the first and the second equations of the system~\eqref{eq:2DSWEsSysOfTemplateFormEqs}
with coefficients~$\lambda^1$ and~$\lambda^2$ that are functions of~$t$ and depend on the parameters involved in elements of~${\rm Ch}^0_b$,
\begin{gather*}
F^1_t(xb_x+yb_y+2b)+c_1(yb_x-xb_y)+F^2b_x+F^3b_y-F^1_{ttt}\frac{x^2+y^2}{2}-F^2_{tt}x-F^3_{tt}y-F^4_t\\
=\sum_{i=1}^2\lambda^i\bigg(a^i_1(xb_x\!+yb_y\!+2b)+a^i_2(yb_x\!-xb_y)+a^i_3b_x\!+a^i_4b_y\!+a^i_5\frac{x^2\!+y^2}{2}+a^i_6x+a^i_7y+a^i_8\bigg).
\end{gather*}

\begin{remark}\label{rem:2DSWEsOnNonproportionalityOfMultipliersOfTemplate-formEqs}
The coefficients~$\lambda^1$ and~$\lambda^2$ cannot be proportional to each other
with the same constant multiplier for all elements of~${\rm Ch}^0_b$.
Indeed, otherwise, i.e., if $\mu_1\lambda^1+\mu_2\lambda^2=0$ 
for some constants~$\mu_1$ and~$\mu_2$ with $(\mu_1,\mu_2)\ne(0,0)$,
there would be no additional conservation-law extensions 
as compared to the more general case of conservation-law extensions with $k=1$,
and the linear combination of the equations~\eqref{eq:2DSWEsSysOfTemplateFormEqs} 
with coefficients~$\mu_2$ and~$-\mu_1$
would play the role of a single template-form equation.
For the same reason, the coefficients~$\lambda^1$ and~$\lambda^2$ 
are simultaneously nonvanishing for some elements of~${\rm Ch}^0_b$.
\end{remark}

Similar to the case $k=1$, we can split the above condition 
with respect to the arbitrary element~$b$, its derivatives~$b_x$ and~$b_y$ and the independent variables~$x$ and~$y$, 
obtaining the system
\begin{gather}\label{gen_restriction}
\begin{split}
&F^1_t=a^1_1\lambda^1+a^2_1\lambda^2,\quad c_1=a^1_2\lambda^1+a^2_2\lambda^2,\quad F^2=a^1_3\lambda^1+a^2_3\lambda^2, \quad F^3=a^1_4\lambda^1+a^2_4\lambda^2,\\
&F^1_{ttt}=-a^1_5\lambda^1-a^2_5\lambda^2,\quad F^2_{tt}=-a^1_6\lambda^1-a^2_6\lambda^2,\quad F^3_{tt}=-a^1_7\lambda^1-a^2_7\lambda^2,\\ 
&F^4_t=-a^1_8\lambda^1-a^2_8\lambda^2.
\end{split}
\end{gather}

Depending on the rank of the submatrix $A_2$, which consists of the first two columns of~$A$ and thus $\rank A_2\leqslant2$, 
the further study of the case $k=2$ splits into three subcases, $\rank A_2=2$, $\rank A_2=1$ and $\rank A_2=0$

\medskip\par\noindent $\boldsymbol{\rank A_2=2.}$
We linearly recombine equations of the system~\eqref{eq:2DSWEsSysOfTemplateFormEqs} 
to set $A_2$ to the $2\times2$ identity matrix or, equivalently, $a^1_1=a^2_2=1$ and $a^1_2=a^2_1=0$.
For further simplification of the system~\eqref{eq:2DSWEsSysOfTemplateFormEqs},
we set $a^1_3=a^1_4=a^1_8=0$ by shifting $x$, $y$ and~$b$, respectively.
The vector fields~$\mathbf v_1$ and~$\mathbf v_2$,
which are associated with the first and the second equations of the reduced system~\eqref{eq:2DSWEsSysOfTemplateFormEqs}, respectively, 
commute in view of the condition~\eqref{eq:CompatibilityCondInTermsOfAssociatedVFs},  
which leads to the following system of algebraic equations for the coefficients~$a^i_j$:
\begin{gather}\label{alg_condition}
a^2_3=a^2_4=0,\quad
a^1_6-a^2_7-2a^2_7=0,\quad
a^1_7+a^2_6+2a^2_6=0,\quad
a^2_5=a^2_8=0.
\end{gather}
The system~\eqref{gen_restriction} then reduces to
\begin{gather}\label{restrict2.2}
\begin{split}
&F^1_t=\lambda^1,\quad c_1=\lambda^2,\quad F^2=0, \quad F^3=0,\quad \lambda^1_{tt}+a^1_5\lambda^1=0,\\
&a^1_6\lambda^1+a^2_6\lambda^2=0,\quad a^1_7\lambda^1+a^2_7\lambda^2=0,\quad F^4_t=0.
\end{split}
\end{gather}
Remark~\ref{rem:2DSWEsOnNonproportionalityOfMultipliersOfTemplate-formEqs} implies that
the sixth and the seventh equations of the system~\eqref{restrict2.2}
are equivalent to $a^1_6=a^2_6=0$ and $a^1_7=a^2_7=0$, respectively.
The system~\eqref{eq:2DSWEsSysOfTemplateFormEqs} reduces in the polar coordinates~$(r,\varphi)$ to
$rb_r+2b+\frac12a^1_5r^2=0$, $b_\varphi=0$, 
with the parameter~$a^1_5$ belonging to $\{0,-4,4\}$ (up to scaling equivalence transformations).
The general solution of the last system is 
\begin{gather*}
b=b_0r^{-2}-\frac{a^1_5}8r^2,
\end{gather*}
where we should assume the integration constant~$b_0$ to be nonvanishing
since the above value of~$b$ with $b_0=0$ is related to the case $k=4$.
We can therefore scale~$b_0$ by an equivalence transformation to $\varepsilon=\pm1$
and obtain, depending on the value of the parameter~$a^1_5$,   
Cases~\ref{CLsCase7a},~\ref{CLsCase7b} and \ref{CLsCase7c} 
of Theorem~\ref{thm:ClassificationOf0thOrderCLsOf2DSWEs1}.

\par\medskip\par\noindent $\boldsymbol{\rank A_2=1.}$
We begin by linearly recombining equations of the system~\eqref{eq:2DSWEsSysOfTemplateFormEqs},
thereby reducing the matrix $A_2$ to
\begin{gather*}
A_2=\begin{pmatrix}
a^1_1&a^1_2\\
0&0
\end{pmatrix}.
\end{gather*}
In other words, $a^2_1=a^2_2=0$ and $(a^1_1,a^1_2)\ne(0,0)$.
Since $\rank A_4=2$, we also get $(a^2_3,a^2_4)\ne (0,0)$.
From the condition~\eqref{eq:CompatibilityCondInTermsOfAssociatedVFs} it follows that $a^1_2=0$ and hence $a^1_1\ne0$.
We can thus set $a^1_1=1$ by re-scaling the first equation,
$a^2_3=1$, $a^2_4=0$ by a rotation equivalence transformation and re-scaling the second equation 
as well as $a^1_3=a^1_4=a^1_8=0$ by shifts of~$x$, $y$ and~$b$.
The commutation relation $[\mathbf v_1,\mathbf v_2]=-\mathbf v_2$, 
which is derived from the condition~\eqref{eq:CompatibilityCondInTermsOfAssociatedVFs}, 
yields the following system of algebraic equations for the remaining coefficients~$a^i_j$:
\begin{gather}\label{rk1_alg_restriction}
a^2_5=a^2_7=0,\quad
4a^2_6=a^1_5,\quad
3a^2_8=a^1_6.
\end{gather}
The system~\eqref{gen_restriction} reduces to
\begin{gather}\label{rk1_split}
\begin{split}
&F^1_t=\lambda^1,\quad c_1=0,\quad F^2=\lambda^2,\quad F^3=0,\\
&F^1_{ttt}=-a^1_5\lambda^1,\quad F^2_{tt}=-a^1_6\lambda^1-a^2_6\lambda^2,\quad
a^1_7\lambda^1=0,\quad F^4_t=-a^2_8\lambda^2.
\end{split}
\end{gather}
After recalling Remark~\ref{rem:2DSWEsOnNonproportionalityOfMultipliersOfTemplate-formEqs},
we get from the eighth equation of the system~\eqref{rk1_split} that $a^1_7=0$.

For $a^1_5\ne0$, shifting $x$ and~$b$ and recombining equations~\eqref{eq:2DSWEsSysOfTemplateFormEqs},
we can set $a^1_6=0$. Then $a^2_8=0$.
The integration of the system~\eqref{eq:2DSWEsSysOfTemplateFormEqs} gives the expression
$b=b_0y^{-2}-\frac18a^1_5r^2$,
where again the integration constant~$b_0$ does not vanish in view of the assumption $k=2$ 
as otherwise this value of the arbitrary element~$b$ is associated with $k=4$.
Scaling equivalence transformations allow us to set $b_0,a^1_5/4\in\{-1,1\}$.
As a result, we derive Cases~\ref{CLsCase8b} and~\ref{CLsCase8c} of Theorem~\ref{thm:ClassificationOf0thOrderCLsOf2DSWEs1} 
for $a^1_5=-4$ and $a^1_5=4$, respectively.

For $a^1_5=0$, we also have $a^2_6=0$. 
Integrating the system~\eqref{eq:2DSWEsSysOfTemplateFormEqs}, we obtain
$b=b_0y^{-2}-a^2_8x$,
where again the integration constant~$b_0$ is nonzero in view of the same argument as above.
Up to scaling equivalence transformations and alternating the signs of $(x,u)$,
we can make $b_0\in\{-1,1\}$ and $a^2_8\in\{0,-1\}$,
which leads to Case~\ref{CLsCase8a}.

\medskip\par\noindent $\boldsymbol{\rank A_2=0.}$
In other words, $a^1_1=a^1_2=a^2_1=a^2_2=0$.
Because $\rank A_4=2$, equations of the system~\eqref{eq:2DSWEsSysOfTemplateFormEqs} can be linearly recombined 
for getting $a^1_3=a^2_4=1$ and $a^1_4=a^2_3=0$.
The compatibility condition~\eqref{eq:CompatibilityCondInTermsOfAssociatedVFs} implies
that the vector fields~$\mathbf v_1$ and~$\mathbf v_2$ 
associated to equations of the reduced system~\eqref{eq:2DSWEsSysOfTemplateFormEqs}
commute, i.e.\ $[\mathbf v_1,\mathbf v_2]=0$. 
This condition is equivalent to the system
\begin{gather}\label{eq:2DSWEsK1RankA20SysForA}
a^1_5=a^2_5=0,\quad
a^1_7=a^2_6.
\end{gather}
Due to the last equation, we can set $a^1_7=a^2_6=0$ by rotation equivalence transformations. 
Then the system~\eqref{eq:2DSWEsSysOfTemplateFormEqs} takes the form 
$b_x+a^1_6x+a^1_8=0$, $b_y+a^2_7y+a^2_8=0$.
Its general solution is 
\begin{gather}\label{eq:2DSWEsK2rA20FormOfB}
b=-\frac{a^1_6}{2}x^2-\frac{a^2_7}{2}y^2-a^1_8x-a^2_8y+b_0,
\end{gather}
where $b_0$ is an integration constant. 
Therefore, $a^1_6\ne a^2_7$ since otherwise 
this form of the arbitrary element~$b$ is related to the value $k=4$, 
which contradicts the supposition $k=2$.
We can set $b_0=0$ modulo equivalence transformations of shifts with respect to~$b$.
Up to the $G^\sim$-equivalence, we can also set $a^1_6=\pm 1$, $a^1_8=0$; $|a^2_7|<1$ if $a^1_6a^2_7>0$;
$a^2_8=0$ if $a^2_7\ne0$; $a^2_8\in\{-1,0\}$ if $a^2_7=0$.
$G^\sim$-inequivalent values of~$b$ of the form~\eqref{eq:2DSWEsK2rA20FormOfB}
jointly with the corresponding spaces of conservation-law characteristics are collected in
Cases~\ref{CLsCase9}--\ref{CLsCase13} of Theorem~\ref{thm:ClassificationOf0thOrderCLsOf2DSWEs1}.

\subsection{More independent template-form equations}

We show below that in fact the case $k=3$ is not possible and thus $k=4$ if $k>2$.

\medskip\par\noindent$\boldsymbol{k=3.}$
The condition $\rank A_4=\rank A=k=3$ necessarily implies that $\rank A_2>0$.

The assumption $\rank A_2=2$ results in a contradiction. 
Indeed, under this assumption we can linearly recombining equations of the system~\eqref{eq:2DSWEsSysOfTemplateFormEqs}
to reduce the submatrix $A_2$ to the form
\[
A_2=\begin{pmatrix}
1&0\\
0&1\\
0&0
\end{pmatrix}.
\]
After the reduction, the projections 
of the vector fields $\mathbf v_1$, $\mathbf v_2$ and $\mathbf v_3$
to the space with the coordinates $(x,y)$ are given by
\begin{gather*}
\hat{\mathbf v}_1=(x+a^1_3)\p_x+(y+a^1_4)\p_y,\quad
\hat{\mathbf v}_2=(y+a^2_3)\p_x-(x-a^2_4)\p_y,\quad
\hat{\mathbf v}_3=a^3_3\p_x+a^3_4\p_y.
\end{gather*}
The condition~\eqref{eq:CompatibilityCondInTermsOfProjectionsOfAssociatedVFs}
requires that the commutator $[\hat{\mathbf v}_2,\hat{\mathbf v}_3]=-a^3_4\p_x+a^3_3\p_y$
belongs to the span $\langle\hat{\mathbf v}_1,\hat{\mathbf v}_2,\hat{\mathbf v}_3\rangle$
but this is not the case.

Hence $\rank A_2=1$, and consequently the linear recombination of equations~\eqref{eq:2DSWEsSysOfTemplateFormEqs}
reduces the matrix $A_4$ to the form
\[
A_4=\begin{pmatrix}
a^1_1&a^1_2&0&0\\
0&0&1&0\\
0&0&0&1
\end{pmatrix},
\quad\mbox{where}\quad (a^1_1,a^1_2)\ne(0,0).
\]
Then the compatibility condition~\eqref{eq:CompatibilityCondInTermsOfAssociatedVFs} 
is expanded to the commutation relations
\[
[\mathbf v_1,\mathbf v_2]=-a^1_1\mathbf v_2+a^1_2\mathbf v_3,\quad
[\mathbf v_1,\mathbf v_3]=-a^1_1\mathbf v_3-a^1_2\mathbf v_2,\quad
[\mathbf v_2,\mathbf v_3]=0.
\]
The last commutation relation implies $a^2_5=a^3_5=0$ and $a^3_6=a^2_7$.
Modulo rotation equivalence transformations, we can set $a^2_7=a^3_6=0$.
The former commutation relations together with the last constraint imply the equations 
$a^1_1(a^2_6-a^3_7)=0$ and $a^1_2(a^2_6-a^3_7)=0$.
In view of $(a^1_1,a^1_2)\ne(0,0)$, we then derive $a^3_7=a^2_6$.
In view of the equations $b_x+a^2_6x+a^2_8=0$ and $b_y+a^2_6y+a^3_8=0$, 
the arbitrary element~$b$ is necessarily a quadratic polynomial of $(x,y)$ 
with the same coefficients of $x^2$ and of $y^2$ and with zero coefficient of~$xy$. 
It is obvious that each such quadratic polynomial satisfies 
a system of the form~\eqref{eq:2DSWEsSysOfTemplateFormEqs} with $k=4$,
which contradicts the assumption $k=3$.

\medskip\par\noindent$\boldsymbol{k=4.}$
Taking into account the condition $\rank A_4=\rank A=4$,
we linearly recombine the equations~\eqref{eq:2DSWEsSysOfTemplateFormEqs} 
to reduce $A_4$ to the $4\times4$ identity matrix.
Due to simultaneously reducing the form of the vector fields~$\mathbf v_1$, \dots, $\mathbf v_4$,
\noprint{
\begin{gather*}
\mathbf v_1=x\p_x+y\p_y-\left(2a^1_1b+\tfrac12a^1_5(x^2+y^2)+a^1_6x+a^1_7y+a^1_8\right)\p_b,\\
\mathbf v_2=y\p_x-x\p_y-\left(2a^2_1b+\tfrac12a^2_5(x^2+y^2)+a^2_6x+a^2_7y+a^2_8\right)\p_b,\\
\mathbf v_3=\p_x-\left(2a^3_1b+\tfrac12a^3_5(x^2+y^2)+a^3_6x+a^3_7y+a^3_8\right)\p_b,\\
\mathbf v_4=\p_y-\left(2a^4_1b+\tfrac12a^4_5(x^2+y^2)+a^4_6x+a^4_7y+a^4_8\right)\p_b,
\end{gather*}
}
the compatibility condition~\eqref{eq:CompatibilityCondInTermsOfAssociatedVFs} in particular gives
the commutation relations
\[
[\mathbf v_2,\mathbf v_3]=\mathbf v_4,\quad
[\mathbf v_2,\mathbf v_4]=-\mathbf v_3,\quad
[\mathbf v_3,\mathbf v_4]=0.
\]
The last commutation relation leads to the equations $a^3_5=a^4_5=0$ and $a^3_7=a^4_6$.
Revisiting the first two commutation relations, we derive
$a^2_5=-2a^3_7=2a^3_7$, $a^3_6=a^4_7$, and thus $a^3_7=0$. 
As a result, the subsystem of the last two equations~\eqref{eq:2DSWEsSysOfTemplateFormEqs}  
takes the form $b_x+a^3_6x+a^3_8=0$, $b_y+a^3_6y+a^4_8=0$.
Any solution of this subsystem is a quadratic polynomial of $(x,y)$ 
with the same coefficients of $x^2$ and of $y^2$ and with zero coefficient of~$xy$. 
Moreover, $k=4$ for all such values of~$b$, 
and only four of them are $G^\sim$-inequivalent,
$b=0$, $b=x$, $b=\tfrac12(x^2+y^2)$ and $b=-\tfrac12(x^2+y^2)$.
which correspond to Cases~\ref{CLsCase14a}, \ref{CLsCase14b}, \ref{CLsCase14c} and~\ref{CLsCase14d}
of Theorem~\ref{thm:ClassificationOf0thOrderCLsOf2DSWEs1}, respectively.

\section*{Acknowledgments}

We thank the two anonymous reviewers for providing valuable suggestions that contributed to the improvement of our paper.
We are also grateful to Lada Atamanchuk-Anhel, Stanislav Opanasenko, Dmytro Popovych, Galyna Popovych and Artur Sergyeyev 
for useful discussions and interesting comments. 
The research of AB was undertaken, in part, thanks to funding from the Canada Research Chairs program,
the InnovateNL LeverageR{\&}D program and the NSERC Discovery program.
The research of ROP was supported by the Austrian Science Fund (FWF), projects P25064 and P30233.

\footnotesize

\end{document}